%% file: p53-arxiv.final.tex
\documentclass[a4paper,USenglish]{lipics-v2018}
\usepackage[utf8]{inputenc}
\usepackage{cite}   
\usepackage[]{todonotes}
\usepackage{graphicx,amssymb,amsmath,url,thm-restate}
\usepackage{xspace}
\usepackage{mathtools}
\usepackage{mathrsfs}
\usepackage{complexity}
\tikzset{font=\small}

\newtheorem{problem}{Problem}

\EventShortTitle{SoCG 2018}
\ArticleNo{53} %

\author{Fabian Klute}{Algorithms and Complexity Group, TU Wien\\
	Vienna, Austria}{fklute@ac.tuwien.ac.at}{}{}
\author{Martin Nöllenburg}{Algorithms and Complexity Group, TU Wien\\
	Vienna, Austria}{noellenburg@ac.tuwien.ac.at}{0000-0003-0454-3937}{}
\authorrunning{F. Klute, M. Nöllenburg}
\title{Minimizing Crossings in Constrained Two-Sided Circular Graph Layouts}
\subjclass{\ccsdesc[500]{Human-centered computing~Graph drawings}, 
\ccsdesc[500]{Mathematics of computing~Graph algorithms}, 
\ccsdesc[300]{Theory of computation~Computational geometry}}
\keywords{Graph Drawing, Circular Layouts, Crossing Minimization, Circle Graphs, Bounded-Degree Maximum-Weight Induced Subgraphs}
\Copyright{Fabian Klute and Martin Nöllenburg}

\relatedversion{This is the full version of a paper with the same title appearing in the proceedings of the 34th International Symposium on Computational Geometry (SoCG) 2018.}

\newcommand{\order}{\ensuremath{ \pi }\xspace}

\newcommand{\intervalset}{\ensuremath{\mathcal{I}}\xspace}
\newcommand{\overlap}{\ensuremath{\mathcal{P}}\xspace}
\newcommand{\iI}{\ensuremath{I}\xspace}
\newcommand{\iJ}{\ensuremath{J}\xspace}
\newcommand{\iK}{\ensuremath{K}\xspace}
\newcommand{\sI}{\ensuremath{a}\xspace}
\newcommand{\eI}{\ensuremath{b}\xspace}
\newcommand{\sJ}{\ensuremath{c}\xspace}
\newcommand{\eJ}{\ensuremath{d}\xspace}
\newcommand{\weight}[1]{\ensuremath{w(#1)}\xspace}
\newcommand{\drawing}{\ensuremath{\Gamma}\xspace}
\newcommand{\tsdrawing}{\ensuremath{\Delta}\xspace}
\newcommand{\inEdges}{\ensuremath{\mathcal E_1}\xspace}
\newcommand{\exEdges}{\ensuremath{\mathcal E_2}\xspace}

\newcommand{\dms}[1]{\ensuremath{#1\text{MWOS}}\xspace}

\DeclareMathOperator{\length}{\ell}
\DeclareMathOperator{\fit}{fit}
\DeclareMathOperator{\Ispan}{span}
\newcommand{\CC}{C\nolinebreak\hspace{-.05em}\raisebox{.4ex}{\tiny\bf +}\nolinebreak\hspace{-.10em}\raisebox{.4ex}{\tiny\bf +}}
\newcommand{\GG}{g\nolinebreak\hspace{-.05em}\raisebox{.4ex}{\tiny\bf +}\nolinebreak\hspace{-.10em}\raisebox{.4ex}{\tiny\bf +}}

\newcommand{\ksubgraph}[1]{{$#1$-BDMWIS}\xspace}
\newcommand{\kiss}[1]{{max-weight $#1$-overlap set}\xspace}
\newcommand{\ksubgraphN}[1]{{\textsc{$#1$-BDMWIS}}\xspace}
\newcommand{\kissN}[1]{{\textsc{Max-Weight $#1$-Overlap Set}}\xspace}

\nolinenumbers

\hideLIPIcs

\begin{document}
	\maketitle
	
	\begin{abstract}
		Circular layouts are a popular graph drawing style, where vertices are placed on a circle and edges are drawn as straight chords. 
		Crossing minimization in circular layouts is \NP-hard. One way to allow for fewer crossings in practice are two-sided layouts that draw some edges as curves in the exterior of the circle. 
		In fact, one- and two-sided circular layouts are equivalent to one-page and two-page book drawings, i.e., graph layouts with all vertices placed on a line (the \emph{spine}) and edges drawn in one or two distinct half-planes (the \emph{pages}) bounded by the spine. 
		In this paper we study the problem of minimizing the crossings for a fixed cyclic vertex order by computing an optimal $k$-plane set of exteriorly drawn edges for $k \ge 1$, extending the previously studied case $k=0$.
		We show that this relates to finding bounded-degree maximum-weight induced subgraphs of circle graphs, which is a graph-theoretic problem of independent interest. 
		We show \NP-hardness for arbitrary $k$, present an efficient algorithm for $k=1$, and generalize it to an explicit \XP-time algorithm for any fixed $k$.
		For the practically interesting case $k=1$ we implemented our algorithm and present experimental results that confirm the applicability of our algorithm.
	\end{abstract}

	\section{Introduction}
		Circular graph layouts are a popular drawing style to visualize graphs, e.g., in biology~\cite{ksbcgh-ciacg-09}, and circular layout algorithms~\cite{st-cda-13} are included in standard graph layout software~\cite{jm-gds-04} such as yFiles, Graphviz, or OGDF. 
		In a circular graph layout all vertices are placed on a circle, while the edges are drawn as straight-line chords of that circle, see Fig.~\ref{fig:interior}. 
		Minimizing the number of crossings between the edges is the main algorithmic problem for optimizing the readability of a circular graph layout. 
		If the edges are drawn as chords, then all crossings are determined solely by the order of the vertices.
		By cutting the circle between any two vertices and straightening it, circular layouts immediately correspond to one-page book drawings, in which all vertices are drawn on a line (the \emph{spine}) and all edges are drawn in one half-plane (the \emph{page}) bounded by the spine.
		Finding a vertex order that minimizes the crossings is \NP-hard~\cite{masuda1987np}. 
		Heuristics and approximation algorithms have been studied in numerous papers, see, e.g.,~\cite{bb-crcl-04,SSSV96,kmn-eebda-17}.

		Gansner and Koren~\cite{gansner2006improved} presented an approach to compute improved circular layouts for a given input graph $\mathcal G = (\mathcal V, \mathcal E)$ in a three-step process. 
		The first step computes a vertex order of $\mathcal V$ that aims to minimize the overall edge length of the drawing, the second step determines a crossing-free subset of edges that are drawn outside the circle to reduce edge crossings in the interior (see Fig.~\ref{fig:auto0}), and the third step introduces edge bundling to save ink and reduce clutter in the interior. 
		The layouts by Gansner and Koren draw edges inside and outside the circle and thus are called \emph{two-sided circular layouts}. 
		Again, it is easy to see that two-sided circular layouts are equivalent to two-page book drawings, where the interior of the circle with its edges corresponds to the first page and the exterior to the second page.

		Inspired by their approach we take a closer look at the second step of the above process, which, in other words, determines for a given cyclic vertex order an outerplane subgraph to be drawn outside the circle such that the remaining crossings of the chords are minimized.		
		Gansner and Koren~\cite{gansner2006improved} solve this problem in $O(|\mathcal V|^3)$ time.\footnote{The paper claims $O(|\mathcal V|^2)$ time without a proof; the immediate running time of their algorithm is~$O(|\mathcal V|^3)$.}	
		In fact, the problem is equivalent to finding a maximum independent set in the corresponding circle graph $G^\circ=(V,E)$, which is the intersection graph of the chords (see Section~\ref{sec:transform} for details). 
		The maximum independent set problem in a circle graph can be solved in $O(\ell)$ time~\cite{valiente2003new}, where $\ell$ is the total chord length of the circle graph (here $|\mathcal E| \le \ell \le |\mathcal E|^2$; see Fig~\ref{sub:1overlap} for a precise definition of $\ell$). 

\paragraph*{Contribution.}
		 We generalize the above crossing minimization problem from finding an outerplane graph to finding an outer $k$-plane graph, i.e., we ask for an edge set to be drawn outside the circle such that none of these edges has more than $k$ crossings. 
		 Equivalently, we ask for a page assignment of the edges in a two-page book drawing, given a fixed vertex order, such that in one of the two pages each edge has at most $k$ crossings. 
		 For $k=0$ this is exactly the same problem considered by Gansner and Koren~\cite{gansner2006improved}. 
		 An example for $k=1$ is shown in Fig.~\ref{fig:2sided}. 
		 More generally, studying drawings of non-planar graphs with a bounded number of crossings per edge is a topic of great interest in graph drawing, see~\cite{KobourovLM17,Liotta2014}.

		\begin{figure}[tb]
			\begin{subfigure}[t]{.28\textwidth}
				\includegraphics{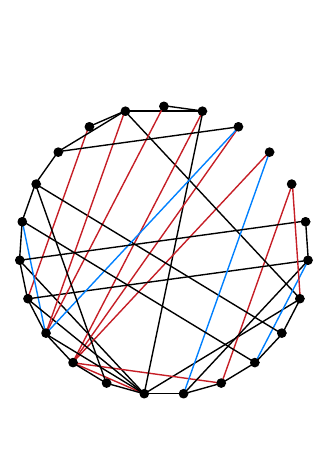}
				\subcaption{One-sided layout with\\ 125 crossings.}
				\label{fig:autoc}\label{fig:interior}
			\end{subfigure}
			\begin{subfigure}[t]{.34\textwidth}
				\includegraphics{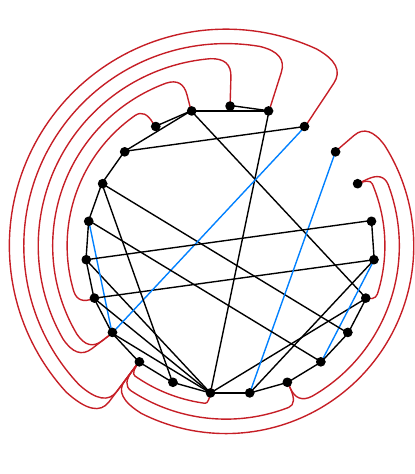}
				\subcaption{Two-sided layout for $ k = 0 $\\ with 48 crossings.}				
				\label{fig:auto0}				
			\end{subfigure}
			\begin{subfigure}[t]{.32\textwidth}
				\includegraphics{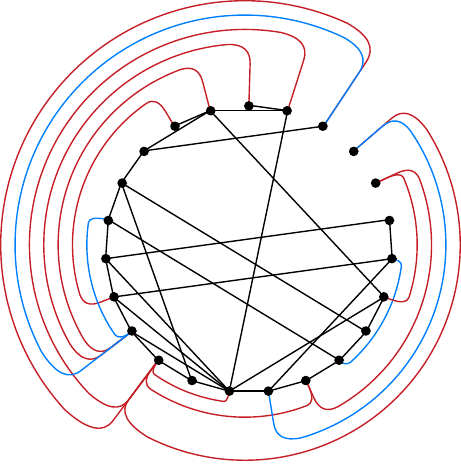}
				\subcaption{Two-sided layout for $ k = 1 $\\ with 30 crossings}				
				\label{fig:auto1}\label{fig:2sided}
			\end{subfigure}
			\caption{Circular layouts of a graph $(\mathcal G, \pi)$ (23 vertices, 45 edges) computed by our algorithms.}
			\label{fig:example}
			\label{fig:auto}			
		\end{figure}

	We model the outer $k$-plane crossing minimization problem in two-sided circular layouts as a bounded-degree maximum-weight induced subgraph (BDMWIS) problem in the corresponding circle graph (Section~\ref{sec:transform}). 
	The BDMWIS problem is a natural generalization of the weighted independent set problem (setting the degree bound $k=0$), which was the basis for Gansner and Koren's approach~\cite{gansner2006improved}.
	It is itself a weighted special case of the bounded-degree vertex deletion problem~\cite{djl-mkfp-93,bbnu-bvdpt-12,ganian2018structural}, a well-studied algorithmic graph problem of independent interest.
 	For arbitrary $k$ we show \NP-hardness of the BDMWIS problem in Section~\ref{sec:hard}. 
	Our algorithms in Section~\ref{sec:algorithms} are based on dynamic programming using interval representations of circle graphs.
	For the case $k=1$, where at most one crossing per exterior edge is permitted, we solve the BDMWIS problem for circle graphs in $O(|\mathcal E|^{4})$ time. %
	We then generalize our algorithm and obtain a problem-specific \XP-time algorithm for circle graphs and any fixed $k$, whose running time is $O(|\mathcal E|^{2k+2})$.
	We note that the pure existence of an \XP-time algorithm can also be derived from applying a metatheorem of Fomin et al.~\cite{fomin2015large} using counting monadic second order (CMSO) logic, but the resulting running times are far worse.
	Finally, in Section~\ref{sec:experiments}, we present the results of a first experimental study comparing the crossing numbers of two-sided circular layouts for the cases $k=0$ and $k=1$.

	\section{Preliminaries}
		\label{sec:transform}

		Let $ \mathcal G = (\mathcal V, \mathcal E) $ be a graph and $ \order $ a cyclic order of $\mathcal V$. 
		We arrange the vertices in  order $\order$ on a circle $C$ and draw edges as straight chords to obtain a (\emph{one-sided}) \emph{circular drawing} \drawing, see Fig.~\ref{fig:interior}. 
		Note that all crossings of \drawing are fully determined by \order: two edges cross iff their endpoints alternate in $\order$. 
		Our goal in this paper is to find a subset of edges to be drawn in the unbounded region outside~$C$ with no more than $k$ crossings per edge in order to minimize the total number of edge crossings or the number of remaining edge crossings inside~$C$.

More precisely, in a \emph{two-sided circular drawing} \tsdrawing of $(\mathcal G,\order)$ we still draw all vertices on a circle $C$ in the order \order, but we split the edges into two disjoint sets $ \inEdges $ and $ \exEdges $ with $\inEdges \cup \exEdges = \mathcal E$. 
The edges in \inEdges are drawn as straight chords, while the edges in \exEdges are drawn as simple curves in the exterior of $C$, see Fig.~\ref{fig:2sided}. 
Asking for a set \exEdges that globally minimizes the crossings in \tsdrawing is equivalent to the \NP-hard fixed linear crossing minimization problem in 2-page book drawings~\cite{mnkf-cmleg-90}.
Hence we add the additional constraint that the exterior drawing induced by \exEdges is \emph{outer $k$-plane}, i.e., each edge in \exEdges is crossed by at most $k$ other edges in \exEdges. 
This is motivated by the fact that, due to their detour, exterior edges are already harder to read and hence should not be further impaired by too many crossings. 
The parameter~$k$, which can be assumed to be small, gives us control on the maximum number of crossings per exterior edge. 
Previous work~\cite{gansner2006improved} is limited to the case $k=0$. 

\subsection{Problem transformation} %
\label{sub:problem_transformation}

		Instead of working with a one-sided input layout $\drawing$ of $( \mathcal G, \order) $ directly we consider the corresponding \emph{circle graph} $ G^\circ = (V, E) $ of $( \mathcal G, \order) $. 
		The vertex set $ V $ of $G^\circ$ has one vertex for each edge in $ \mathcal E $ and two vertices $ u,v \in V $ are connected by an edge $(u,v)$ in $E$ if and only if the chords corresponding to $u$ and $v$ cross in  \drawing, i.e., their endpoints alternate in $\order$.
		The number of vertices $|V|$ of $G^\circ$ thus equals the number of edges $|\mathcal E|$ of $\mathcal G$ and the number of edges $|E|$ of $G^\circ$ equals the number of crossings in \drawing. 
		Moreover, the degree $\deg(v)$ of a vertex $v$ in $G^\circ$ is the number of crossings of the corresponding edge in $\drawing$.

		Next we show that we can reduce our outer $k$-plane crossing minimization problem in two-sided circular layouts of $(\mathcal G, \order)$ to an instance of the following bounded-degree maximum-weight induced subgraph problem for $G^\circ$. 
		\begin{problem}[Bounded-Degree $k$ Maximum-Weight Induced Subgraph (\ksubgraphN{k})]
			\label{def:ksubgraphN}
			Let  $ G = (V,E) $ be a weighted graph with a vertex weight $ \weight{v} \in \mathbb{R}^+ $ for each $v \in V$ and an edge weight $ \weight{u,v} \in \mathbb{R}^+ $ for each $(u,v) \in E$ and let $ k \in \mathbb{N} $. 
			Find a set $ V' \subset V $ such that the induced subgraph $ G[V'] = (V',E') $ has maximum vertex degree $ k $ and maximizes the weight \[ W = W(G[V']) = \sum_{v \in V'}\weight{v} - \sum_{(u,v) \in E'}\weight{u,v}. \]
		\end{problem}
		
		For general graphs it follows immediately from Yannakakis~\cite{yannakakis1978node} that \ksubgraphN{k} is \NP-hard, but restricting the graph class to circle graphs makes the problem significantly easier, at least for constant $k$, as we  show in this paper. 

		For our reduction it remains to assign suitable vertex and edge weights to $G^\circ$.
		We define $w(v) = \deg(v)$ for all vertices $v \in V$ and $w(u,v) = 1$ or, alternatively, $w(u,v) = 2$ for all edges $(u,v) \in E$, depending on the type of crossings to minimize. 
		
		\begin{lemma}
		\label{lem:ma1indset}
		Let $ \mathcal G = (\mathcal V, \mathcal E) $ be a graph with cyclic vertex order $ \order $ and $ k\in \mathbb{N} $. Then a maximum-weight degree-$k$ induced subgraph in the corresponding weighted circle graph $ G^\circ = (V,E)$ induces an outer $k$-plane graph that minimizes the number of crossings in the corresponding two-sided layout \tsdrawing of $(\mathcal G, \order)$.
		\end{lemma}
		\begin{proof}
			Let $V^* \subset V$ be a vertex set that induces a maximum-weight subgraph of degree at most $k$ in $G^\circ$. 
			Since vertices in $G^\circ$ correspond to edges in $\mathcal G$, we can choose $\mathcal E^* = V^*$ as the set of exterior edges in \tsdrawing.
			Each edge in $G^\circ$ corresponds to a crossing in the one-sided circular layout \drawing. 
			Hence each edge in the induced graph $G^\circ[V^*]$ corresponds to an exterior crossing in \tsdrawing. 
			Since the maximum degree of $G^\circ[V^*]$ is $k$, no exterior edge in \tsdrawing has more than $k$ crossings. 
									
			The degree of a vertex $v \in V^*$ (and thus its weight $w(v)$) equals the number of crossings that are removed from \drawing by drawing the corresponding edge in $\mathcal E^*$ in the exterior part of \tsdrawing. 
			However, if two vertices in $V^*$ are connected by an edge, their corresponding edges in $\mathcal E^*$ necessarily cross in the exterior part of \tsdrawing and we need to add a correction term, otherwise the crossing would be counted twice. 
			So for edge weights $w(u,v)=1$ the weight $W$ maximized by $V^*$ equals the number of crossings that are removed from the interior part of~\tsdrawing. 
			For $w(u,v)=2$, though, the weight $W$ equals the number of crossings that are removed from the interior, but not counting those that are simply shifted to the exterior of~\tsdrawing.
		\end{proof}
		
		Lemma~\ref{lem:ma1indset} tells us that instead of minimizing the crossings in two-sided circular layouts with an outer $k$-plane exterior graph, we can focus on solving the \ksubgraphN{k} problem for circle graphs in the rest of the paper. 

\subsection{Interval representation of circle graphs} %
\label{sub:overlap_graphs}
	There are two alternative representations of circle graphs. 
	The first one is the \emph{chord representation} as a set of chords of a circle (i.e., a one-sided circular layout), whose intersection graph actually serves as the very definition of circle graphs. 
	The second and less immediate representation is the \emph{interval representation} as an \emph{overlap graph}, which is more convenient for describing our algorithms. 
	In an interval representation each vertex is represented as a closed interval $I \subset \mathbb R$.
	Two vertices are adjacent if and only if the two corresponding intervals partially overlap, i.e., they intersect but neither interval contains the other.

	\begin{figure}[htb]
		\includegraphics{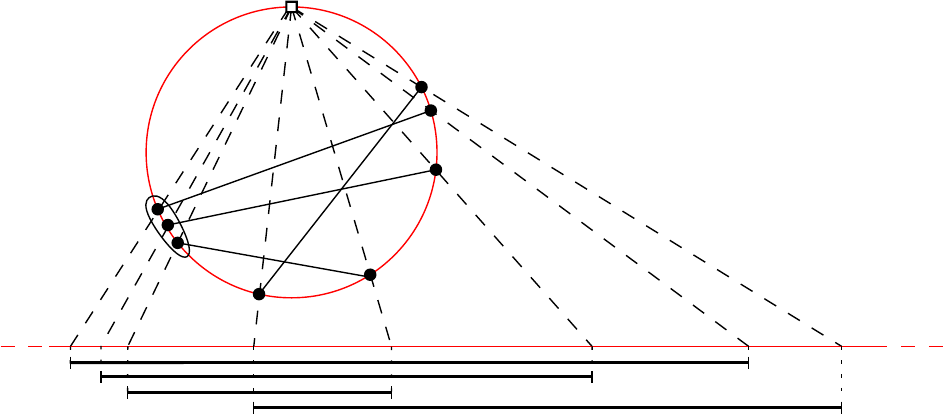}
		\centering
		\caption{An example projection of the chord representation of a circle graph (here: $K_{1,3}$) to obtain an interval representation of the same graph as an overlap graph. Marked groups of endpoints indicate how chords incident to the same vertex are separated before the projection.}
		\label{fig:proj_exmp}
	\end{figure}

	Gavril~\cite{gavril1973algorithms} showed that circle graphs and overlap graphs represent the same class of graphs. 
	To obtain an interval representation from a chord representation $\drawing$ on a circle $C$ the idea is to pick a point $p$ on $C$, which is not the endpoint of a chord, rotate $\Gamma$ such that $p$ is the topmost point of $\Gamma$ and project the chords from $p$ onto the real line below $C$, see Fig.~\ref{fig:proj_exmp}.  
	Each chord is then represented as a finite interval and two chords intersect if and only if their projected intervals partially overlap.
	We can further assume that all endpoints of the projected intervals are distinct by locally separating chords with a shared endpoint in $\drawing$ before the projection, such that the intersection graph of the chords does not change. %

	\section{\NP-hardness}\label{sec:hard}
		For arbitrary, non-constant $ k \in \mathbb{N} $ we show that \ksubgraphN{k} is \NP-hard, even on circle graphs. %
		Our reduction is from the \textsc{Minimum Dominating Set} problem, which is \NP-hard on circle graphs~\cite{keil1993complexity}.
		
		\begin{problem}[\textsc{Minimum Dominating Set}]
			Given a graph $ G = (V,E) $, find a set $ V' \subseteq V $ of minimum cardinality such that for each $u \in V \setminus V'$ there is a vertex $v \in V'$ with $(u,v) \in E$.
		\end{problem}

		\begin{theorem}
			\label{thm:hard}
			\ksubgraphN{k} is \NP-hard on circle graphs, even if 
			all vertex weights are one and all edge weights are zero.
		\end{theorem}
		\begin{proof}
			Given an instance of \textsc{Minimum Dominating Set} on a circle graph $ G = (V,E) $ we construct an instance of \ksubgraphN{k}. 
			First let $ G' = (V', E') $ be a copy of $ G $. 
			We set the degree bound $ k $ equal to the maximum degree of $ G $ and attach new leaves to each vertex $ v \in V' $ until every (non-leaf) vertex in $ G' $ has degree $ k + 1 $. 
			Note that $G'$ remains a circle graph when adding leaves. 
			We set the weights to $ w(v) = 1 $ for $ v \in V' $ and $ w(u,v) = 0 $ for $ (u,v) \in E' $. 
			This implies that the weight $W$ to be maximized is just the number of vertices in the induced subgraph.
			
			Now given a minimum dominating set $ V_d \subseteq V $ of $ G $,  we know that for every vertex $ v \in V $ either $ v \in V_d $ or there exists a vertex $ u \in V_d $ such that $ (u,v) \in E $. 
			This means if we set $ V_s = V' \setminus V_d $ the graph $ G'[V_s] $ has maximum degree $ k $, since for every $ v \in V_s $ at least one neighbor is in $V_d$ and the maximum degree in $G'$ is $k+1$. 
			Since $ V_d $ is a minimum dominating set,  $ V_s $, for which we can assume that it contains all leaves, is the largest set of vertices such that $ G'[V_s] $ has maximum degree $ k $. 
			Hence $ V_s $ is a solution to the \ksubgraphN{k} problem on~$ G' $.
			
			Conversely let $ V_s \subseteq V' $ be a solution to the \ksubgraphN{k} problem on $ G' $.
			Again we can assume that $V_s$ contains all leaves of $G'$. 
			Otherwise let $u \in V'\setminus V_s$ be a leaf of $G'$ with unique neighbor $v \in V'$. 
			The only possible reason that $u \not\in V_s$ is that $v \in V_s$ and the degree of $v$ in $G'[V_s]$ is $\deg(v) = k$. 
			If we replace in $V_s$ a non-leaf neighbor $w$ of $v$ (which must exist) by the leaf $u$, the resulting set has the same cardinality and satisfies the degree constraint.			
			 Now let $ V_d = V' \setminus V_s $. 
			 By our assumption $V_d$ contains no leaves of $G'$ and $V_d \subseteq V$. 
			 Since every vertex in $ G'[V_s] $ has degree at most $ k $ we know that each $v \in V \setminus V_d $ must have a neighbor $u \in V_d$, otherwise it would have degree $k+1$ in $G'[V_s]$.  
			 Thus $V_d$ is a dominating set.
			 Further, $V_d$ is a minimum dominating set. 
			 If there was a smaller dominating set $V_d'$ in $G$ then $V' \setminus V_d'$ would be a larger solution than $V_s$ for the \ksubgraphN{k} problem on $ G' $, which is a contradiction.
		\end{proof}

\section{Algorithms for \ksubgraphN{k} on circle graphs }\label{sec:algorithms}
		Before describing our dynamic programming algorithms for $k=1$ and the generalization to $k \ge 2$ in this section, we introduce the necessary basic definitions and notation using the interval perspective on \ksubgraphN{k} for circle graphs.

		\subsection{Notation and Definitions} %
		\label{sub:notation_and_definitions}
		Let $G=(V,E)$ be a circle graph and $\intervalset = \{I_1, \dots, I_n\}$ an interval representation of $G$ with $n$ intervals that have $2n$ distinct endpoints as defined in Section~\ref{sub:overlap_graphs}.
		Let $\sigma(\intervalset) = \{\sigma_1, \dots, \sigma_{2n}\}$ be the set of all interval endpoints and assume that they are sorted in increasing order, i.e., $\sigma_i < \sigma_j$ for all $i < j$.
		We may in fact assume without loss of generality that $\sigma(\intervalset) = \{1, \dots, 2n\}$ by mapping each interval $[\sigma_l,\sigma_r] \in \intervalset$ to the interval $[l,r]$ defined by its index ranks.
		Clearly the order of the endpoints of two intervals $[\sigma_l,\sigma_r]$ and $[\sigma_{l'},\sigma_{r'}]$ is exactly the same as the order of the endpoint of the intervals $[l,r]$ and $[l',r']$ and thus the overlap or circle graph defined by the new interval set is exactly the same as the one defined by $\intervalset$.
		
		For two distinct intervals $\iI=[a,b]$ and $ \iJ=[c,d] \in \intervalset$ we say that $\iI$ and $\iJ$ \emph{overlap} if $a<c<b<d$ or $c<a<d<b$.
		Two overlapping intervals correspond to an edge in $G$.
		For an interval $I \in \intervalset$ and a subset $\mathcal I' \subseteq \intervalset$ we define the overlap set $\overlap(I, \intervalset') = \{J \mid J \in \intervalset' \text{ and } \iI,\iJ \text{ overlap} \}$. 
		Further, for $\iI=[a,b]$, we define the forward overlap set $ \overrightarrow{\overlap}(\iI,\intervalset') =  \{J \mid J=[c,d] \in \overlap(I, \intervalset') \text{ and } \sJ < \eI < \eJ \}$ of intervals overlapping on the right side of $I$ and the set $\overlap(\intervalset') = \{\{\iI,\iJ\} \mid \iI, \iJ \in \intervalset' \text{ and } \iJ \in \overlap(I, \intervalset')\}$ of all overlapping pairs of intervals in $\intervalset'$.
		If $\iJ \subset \iI$, i.e., $a<c<d<b$, we say that $\iI$ \emph{nests} $\iJ$ (or $\iJ$ is \emph{nested} in $\iI$).
		Nested intervals do not correspond to edges in $G$.
		For a subset $\mathcal I' \subseteq \intervalset$ we define the set of all intervals nested in $I$ as $\mathcal{N}(I,\intervalset') = \{J \mid J \in \intervalset' \text{ and } \iJ \text{ is nested in } \iI\}$.

		Let $ \intervalset' \subseteq \intervalset $ be a set of $ n' $ intervals. 
		We say $ \intervalset' $ is \emph{connected} if its corresponding overlap or circle graph is connected. 
		Further let $ \sigma(\intervalset') = \{i_1, \dots, i_{2n'} \}$ be the sorted interval endpoints of $ \intervalset' $. 
		The \emph{span} of $ \intervalset' $ is defined as $ \Ispan(\intervalset') = i_{2n'} - i_1 $ and the \emph{fit} of the set $ \intervalset' $ is defined as $ \fit(\intervalset') = \max\{i_{j+1} - i_j \mid 1 \leq j < 2n' \} $. %

		For a weighted circle graph $G=(V,E)$ with interval representation $\intervalset$ we can immediately assign each vertex weight $w(v)$ as an interval weight $w(I_v)$ to the interval $I_v \in \intervalset$ representing $v$ and each edge weight $w(u,v)$ to the overlapping pair of intervals $\{I_u,I_v\} \in \overlap(\intervalset)$ that represents the edge $(u,v) \in E$. 
		We can now phrase the \ksubgraphN{k} problem for a circle graph in terms of its interval representation, i.e., given an interval representation $\intervalset$ of a circle graph $G$, find a subset $\intervalset' \subseteq \intervalset$ such that no $\iI \in \intervalset'$ overlaps more than $k$ other intervals in $\intervalset'$ and such that the weight $W(\intervalset') = \sum_{\iI \in \intervalset'} \weight{\iI} - \sum_{\{\iI,\iJ\} \in \mathcal{P}(\intervalset')} \weight{\iI,\iJ}$ is maximized. 
		We call such an optimal subset $\intervalset'$ of intervals a \emph{max-weight $k$-overlap set}.

		\subsection{Properties of max-weight 1-overlap sets}\label{sub:1overlap}
		
		The basic idea for our dynamic programming algorithm is to decompose any \emph{$1$-overlap set}, i.e., a set of intervals, in which no interval overlaps more than one other interval, into a sequence of independent single intervals and overlapping interval pairs. 
		Consequently, we can find a max-weight 1-overlap set by optimizing over all possible ways to select a single interval or an overlapping interval pair and recursively solving the induced independent subinstances that are obtained by splitting the instance according to the selected interval(s).

		Let $\intervalset$ be a set of intervals.
		For $ x, y \in \mathbb{R} \cup \{\pm \infty\} $ with $x \le y$ we define the set $ \intervalset[x,y] = \{\iI \in \intervalset \mid \iI \subseteq [x,y]\}$ as the subset of $\intervalset$ contained in $[x,y]$. 
		Note that $\intervalset[-\infty,\infty] = \intervalset$.
			For any $\iI  = [a,b]  \in \intervalset[x,y]$ we can \emph{split} $\intervalset[x,y]$ \emph{along}  $ \iI $ into the three sets 
			$ \intervalset[x, \sI]$, $\intervalset[\sI,\eI]$, $\intervalset[\eI, y]$.
			This split corresponds to selecting $\iI$ as an interval without overlaps in a candidate 1-overlap set.
			All intervals which are not contained in one of the three sets will be discarded after the split.

			Similarly, we can split  $\intervalset[x,y]$ %
			along a pair of overlapping intervals $ \iI=[\sI,\eI], \iJ=[\sJ,\eJ] \in \intervalset $ to be included in candidate solution. 
			Without loss of generality let $ \sI < \sJ < \eI < \eJ $. 
			Then the split creates the five sets
			$\intervalset[x, \sI]$ , $\intervalset[\sI,\sJ]$,  $\intervalset[\sJ, \eI]$, $\intervalset[\eI, \eJ]$, $\intervalset[\eJ, y]$, see Fig.~\ref{fig:fivecases}.  
			Again, all intervals which are not contained in one of the five sets are discarded. 
			The next lemma shows that none of the discarded overlapping intervals can be included in a 1-overlap set together with $\iI$ and $\iJ$.

			\begin{figure}[tbp]
				\centering
				\includegraphics[width=\textwidth]{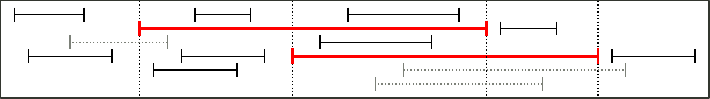}
				\caption{Split along the two thick red intervals. The dotted intervals are discarded and we recurse on the five sets with black intervals.}
				\label{fig:fivecases}
			\end{figure}

			\begin{lemma}
				\label{lem:picktwo}
				For any $ x \in \mathbb{R} $ at most two overlapping intervals $ \iI = [a,b],\iJ = [c,d] \in \intervalset $ with  $ a \leq x \leq b $ and $ c \leq x \leq d $ can be part of a $1$-overlap set of \intervalset.
			\end{lemma}
			\begin{proof}
			Assume there is a third interval $ \iK = [e,f] \in \intervalset $ with $ e \leq x \leq f $ in a $1$-overlap set, which  overlaps \iI or \iJ or both. Interval \iK cannot be added to the 1-overlap set without creating at least one interval that overlaps two other intervals, which is not allowed in a 1-overlap set.
			\end{proof}

			Our algorithm for the \kiss{1} problem extends some of the ideas of the algorithm presented by Valiente for the independent set problem in circle graphs~\cite{valiente2003new}. 
			In our analysis we use Valiente's notion of \emph{total chord length}, where the chord length is the same as the \emph{length} $ \length(\iI) = j - i $ of the corresponding interval $\iI=[i,j] \in \intervalset$. 
			The \emph{total interval length} can then be defined as 
			$\ell = \ell(\intervalset) = \sum_{\iI \in \intervalset} \length(\iI)$. 
			We use the following bound in our analysis.
			\begin{lemma}
				\label{lem:length}
				Let \intervalset be a set of intervals and $ \gamma $ be the maximum degree of the corresponding overlap or circle graph, then 
				$\sum_{\iI\in\intervalset}\sum_{\iJ \in \mathcal{P}(\iI,\intervalset)} (\length(\iI) + \length(\iJ)) = O(\gamma\ell)$. 
			\end{lemma}
			\begin{proof}
				We first observe that $ \iJ \in \mathcal{P}(\iI,\intervalset) $ if and only if $ \iI \in \mathcal{P}(\iJ,\intervalset) $. So in total each interval in $\mathcal I$ appears at most $\gamma$ times as $\iI$ and at most $\gamma$ times as $ \iJ $ in the double sum, i.e., no interval in $\mathcal I$ appears more than $2\gamma$ times and the bound follows.
			\end{proof}

		\subsection{An algorithm for max-weight 1-overlap sets}
			\label{sec:algo}
			
			Our algorithm to compute max-weight 1-overlap sets runs in two phases. 
			In the first phase, we compute the weights of optimal solutions on subinstances of increasing size by recursively re-using solutions of smaller subinstances.  
			In the second phase we optimize over all ways of combining optimal subsolutions to obtain a max-weight 1-overlap set.

			The subinstances of interest are defined as follows. 
			Let $\intervalset' \subseteq \intervalset$ be a connected set of intervals and let $l=l(\intervalset')$ and $r=r(\intervalset')$ be the leftmost and rightmost endpoints of all intervals in $\intervalset'$.
			We define the value $\dms{1}(\intervalset')$ as the maximum weight of a 1-overlap set on $\intervalset[l,r]$ that includes $\intervalset'$ in the 1-overlap set (if one exists).
			Lemma~\ref{lem:picktwo} implies that it is sufficient to compute the $ \dms{1} $ values for single intervals $ \iI \in \intervalset $ and overlapping pairs $ \iI,\iJ \in \intervalset $ since any connected set of three or more intervals cannot be a 1-overlap set any more.

			We start with the computation of $ \dms{1}(\iI) $ for a single interval $\iI = [a,b] \in \intervalset $. 
			Using a recursive computation scheme of $\dms{1}$ that uses increasing interval lengths we may assume by induction that for any interval $ \iJ \in \intervalset $ with $ \length(\iJ) < \length(\iI) $ and any overlapping pair of intervals $ \iJ,\iK \in \intervalset $ with  $ \Ispan(\iJ,\iK) < \length(\iI) $ the sets $ \dms{1}(\iJ) $ and $ \dms{1}(\iJ,\iK) $ are already computed. 
			If we select $\iI$ for the 1-overlap set as a single interval without overlaps, we need to consider for $\dms{1}(\iI)$ only those intervals nested in $\iI$. Refer to Fig.~\ref{fig:rec} for an illustration.
			The value of $ \dms{1}(I) $ is determined using an auxiliary recurrence $ S_I[x]$ for $ a \le x \le b $ and the weight $w(\iI)$ resulting from the choice of $\iI$:
			\begin{align}\label{rec:I}
			\dms{1}([a,b]) = S_I[a+1] + w(\iI).
			\end{align}
			For a fixed interval $\iI= [a,b]$ the value $S_I[x]$ represents the weight of an optimal solution of $\intervalset[x,b]$. To simplify the definition of recurrence $ S_I[x] $ we define the set $ D_S([c,d],\intervalset[a,b]) $ with $ [c,d] \in \intervalset[a,b] $, in which we collect all \dms{1} values for pairs composed of $[c,d]$ and an interval in $\overrightarrow{\mathcal P}([c,d],\intervalset[a,b])$ (see Fig.~\ref{fig:rec}(c)) as
			\begin{align}
				\label{eq:DS}
				D_S([c,d],\intervalset[a,b]) = 
				\{ \dms{1}([c,d],[e,f]) + S_I[f + 1] \mid [e,f] \in \overrightarrow{\mathcal{P}}([c,d],\intervalset[a,b])\}.
			\end{align}
			The main idea of the definition  of $ S_I[x] $ is a maximization step over the already computed sub-solutions that may be composed to an optimal solution for $\intervalset[x,b]$. To stop the recurrence we set $ S_I[b] = 0 $ and for every end-point $ d $ of an interval $ [c,d] \in \intervalset[a,b] $ we set $ S_I[d] = S_I[d+1] $. It remains to define the recurrence for the start-point $c$ of each interval $[c,d]  \in \intervalset[a,b] $:
			\begin{align}
				\label{rec:S}
				S_I[c] = 	\max \{	\{S_I[c+1], \dms{1}([c,d]) + S_I[d + 1]\} \cup D_S([c,d],\intervalset[a,b]) \}.
			\end{align}
			\begin{figure}[t]
				\centering
				\includegraphics[width=\textwidth]{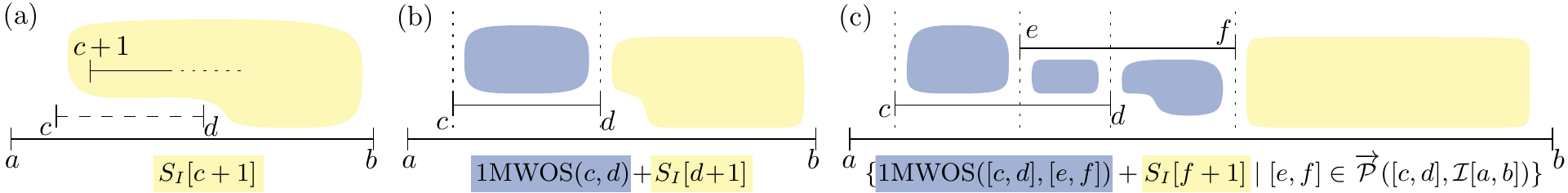}
				\caption{Illustration of Recurrence (\ref{rec:S}). The dashed intervals are discarded, while solid ones are considered in the solution.}
				\label{fig:rec}
			\end{figure}
			Figure~\ref{fig:rec} depicts which of the possible configurations of selected intervals is represented by which values in the maximization step of Recurrence~\ref{rec:S}. The first option (Fig.~\ref{fig:rec}(a)) is to discard the interval $[c,d]$, the second option (Fig.~\ref{fig:rec}(b)) is to select $[c,d]$ as a single interval, and the third option (Fig.~\ref{fig:rec}(c)) is to select $[c,d]$ and an interval in its forward overlap set.
			\begin{lemma}
				\label{lem:dmsI}
				Let $ \intervalset $ be a set of intervals and $ \iI \in \intervalset $, then the  value $ \dms{1}(\iI) $ can be computed in $ O(\gamma\length(\iI)) $ time assuming all $ \dms{1}(\iJ) $ and $ \dms{1}(\iJ,\iK) $ values are computed for $ \iJ,\iK \in \intervalset $, $ \length(\iJ) < \length(\iI) $ and $ \Ispan(\iJ,\iK) < \length(\iI) $.
			\end{lemma}
			\begin{proof}
				Recurrence (\ref{rec:I}) is correct if $ S[a+1] $ is exactly the weight of a max-weight 1-overlap set on the set $ \mathcal{N}(\iI,\intervalset) $, the set of nested intervals of $ \iI $. The proof is by induction over the number of intervals in $ \mathcal{N}(\iI,\intervalset) $. In case $  \mathcal{N}(\iI,\intervalset) $ is empty $ b = a+1 $ and  with $ S[b]=0 $ Recurrence (\ref{rec:I}) is correct. 

				Now let $ \mathcal{N}(\iI,\intervalset) $ consist of one or more intervals. By Lemma~\ref{lem:picktwo} there can only be three cases of how an interval $ \iJ \in \mathcal{N}(\iI,\intervalset) $ contributes. We can decide to discard $\iJ$, to add it as a singleton interval which allows us to split $ \mathcal{N}(\iI,\intervalset) $ along $\iJ$ or to add an overlapping pair $ \iJ, \iK \in \mathcal{N}(\iI,\intervalset) $ such that $ \iK \in \overrightarrow{\mathcal{P}}(\iJ,\intervalset[a,b]) $ and split $\mathcal{N}(\iI,\intervalset)$ along $\iJ,\iK$.
				
				For the start-point $c$ of an interval $\iJ=[c,d]$ the maximization in the definition of $ S $ in Recurrence~(\ref{rec:S}) exactly considers these three possibilities (recall Fig.~\ref{fig:rec}). For all end-points aside from $ b $ we simply use the value of the next start-point or $ S[b]=0 $ which ends the recurrence. Since all $ \dms{1}(\iJ) $ and $ \dms{1}(\iJ,\iK) $ are computed for $ \iJ,\iK \in \intervalset $, $ \length(\iJ) < \length(\iI) $ and $ \Ispan(\iJ,\iK) < \length(\iI) $ the auxiliary table $S$ is computed in one iteration across $ \sigma(\intervalset[a,b]) $. 
				
				The overall running time is dominated by traversing the $ D_S $ sets, which contain at most $ \gamma $ values. This has to be done for every start-point of an interval in $ \mathcal{N}(\iI, \intervalset) $ which leads to an overall computation  time of $ O(\gamma\length(\iI)) $ for $ \dms{1}(\iI) $.
			\end{proof}
			
			Until now we only considered computing the \dms{1} value of a single interval, but we still need to compute \dms{1} for pairs of overlapping intervals. 
			Let $ \iI = [c,d], \iJ = [e,f] \in \intervalset $ be two intervals such that $ \iJ \in \overrightarrow{\mathcal{P}}(\iI,\intervalset) $. 
			If we split \intervalset along these two intervals we find three independent regions (recall Fig.~\ref{fig:rec}(c)) and 
			obtain 
			\begin{align}\label{rec:IJ}
				\dms{1}(\iI,\iJ) = L_{I,J}[c+1] + M_{I,J}[e+1] + R_{I,J}[d+1] + w(\iI) + w(\iJ) - w(\iI,\iJ).
			\end{align}
			The auxiliary recurrences $ L_{I,J},M_{I,J},R_{I,J} $ are defined for the three independent regions in the very same way as $ S_I $ above with the exception that $ L_{I,J}[e] = 0, M_{I,J}[d] = 0 $ and $ R_{I,J}[f] = 0 $. %
			Hence, following essentially the same proof as in Lemma~\ref{lem:dmsI} we obtain
			\begin{lemma}
				\label{lem:timedms}
				Let $ \intervalset $ be a set of intervals and $ \iI,\iJ \in \mathcal I$  with $\iJ \in \overrightarrow{\mathcal P}(\iI,\mathcal I)$, then $ \dms{1}(\iI,\iJ) $ can be computed in $ O(\gamma\Ispan(\iI,\iJ)) $ time assuming all $ \dms{1}(\iK) $ and $ \dms{1}(\iK,L) $ values are computed for $ \iK,L \in \intervalset $, $ \length(\iK) < \fit(\iI,\iJ) $ and $ \Ispan(\iK,L) < \fit(\iI,\iJ) $.
			\end{lemma}

			\begin{lemma}
				\label{lem:dms}
				Let $ \intervalset $ be a set of intervals. The \dms{1} values for all $ \iI \in\intervalset $ and all pairs $ \iI,\iJ \in \intervalset $ with $ \iJ \in \overrightarrow{\mathcal{P}}(\iI,\intervalset) $ can be computed in $ O(\gamma^2\ell) $ time.
			\end{lemma}
			\begin{proof}
				For an interval $ \iI\in \intervalset $ the value  $ \dms{1}(\iI) $ is computed in $ O(\gamma\length(I)) $ time by Lemma~\ref{lem:dmsI}. With $\ell = \sum_{\iI \in \intervalset} \length(\iI) $ the claim follows for all $\iI \in \mathcal I$.

				By Lemma~\ref{lem:timedms} the value $ \dms{1}(\iI,\iJ) $ can be computed in $ O(\gamma\Ispan(\iI,\iJ)) $ time for each overlapping pair $ \iI,\iJ $ with $ \iJ \in \overrightarrow{\mathcal{P}}(\iI,\intervalset) $. Since $\Ispan(\iI,\iJ) \le \length(\iI) + \length(\iJ)$ the time bound of $O(\gamma^2\ell)$ follows by applying Lemma~\ref{lem:length}.				
			\end{proof}

			In the second phase of our algorithm we compute the maximum weight of a 1-overlap set for \intervalset by defining another recurrence $ T[x] $ for $ x \in \sigma(\intervalset) $ and re-using the \dms{1} values.
			The recurrence for $T$ is defined similarly to the recurrence of $S_I$ above. %
			We set $ T[2n] = 0 $. Let $ \iI = [a,b] \in \intervalset $ be an interval and $b \ne 2n$, then
			\begin{align}
				\label{rec:T}
					T[b] = T[b+1]
				&&
					T[a] =
					\max \left\{\begin{tabular}{@{}l@{}}
					$\{T[a+1]\}\ \cup$ \\
					$\{\dms{1}([a,b]) + T[b + 1]\}\ \cup$ \\ 
					$D_T(\iI,\intervalset)$
					\end{tabular}\right\},
			\end{align}
			where $D_T$ is defined analogously to $D_S$ by replacing the recurrence $S_I$ with $T$ in Equation~\ref{eq:DS}.
			The maximum weight of a 1-overlap set for \intervalset is found in $ T[1] $. 
			\begin{theorem}\label{thm:intervals}
				A max-weight 1-overlap set for a set of intervals \intervalset can be computed in $O(\gamma^2\ell) \subseteq O(|\intervalset|^4)$ time, where $\ell$ is the total interval length and $\gamma$ is the maximum degree of the corresponding overlap graph.
			\end{theorem}
			
			 \begin{proof}
			 	The time to compute all \dms{1} values is $ O(\gamma^2\ell) $ with Lemma~\ref{lem:dms}. As argued the optimal solution is found by computing $ T[1] $. 
				The time to compute $ T $ in Recurrence~(\ref{rec:T}) is again dominated by the maximization, which itself is dominated by the evaluation of the $ D_T $ sets. The size of these sets is exactly $ \gamma $ times the sum we bounded in Lemma~\ref{lem:length}. So the total time to compute $T$ is $ O(\gamma^2\ell) $.
				 Hence the total running time is $ O(\gamma^2\ell) $. From $\gamma \le |\intervalset|$ and $\ell \le |\intervalset|^2$ we obtain the coarser bound $O(|\intervalset|^4)$.
			
			 	It remains to show the correctness of Recurrence~(\ref{rec:T}). 
				Again we can treat it with the same induction used in the proof of Lemma~\ref{lem:dmsI}. To see this we introduce an interval $ [0,2n+1] $ with weight zero. 
				Now the computation of the maximum weight of a 1-overlap set for $ \intervalset $ is the same as computing all $ \dms{1}$ values for the instance  $\intervalset \cup \{[0,2n+1]\} $.
				Using standard backtracking, the same algorithm can be used to compute the max-weight 1-overlap set instead of only its weight.				
			 \end{proof}

\subsection{An \XP-algorithm for max-weight $k$-overlap sets}\label{sec:general}

		In this section we generalize our algorithm to $ k \ge 2 $. 
		While it is not possible to directly generalize Recurrences~(\ref{rec:S}) and~(\ref{rec:T}) we do use similar concepts. 
		The difficulty for $ k > 1 $ is that the solution can have arbitrarily large connected parts, e.g., a 2-overlap set can include arbitrarily long paths and cycles. 
		So we can no longer partition an instance along connected components into a constant number of independent sub-instances as we did for the case $k=1$. 
		Due to space constraints we only sketch the main ideas here and refer the reader to Appendix~\ref{apx:XP} for all omitted proofs.
		
		We first generalize the definition of $\dms{1}$. 
		Let $\intervalset$ be a set of $n$ intervals as before and $I=[a,b] \in \intervalset$.
		We define the value $\dms{k}(I)$ as the maximum weight of a $k$-overlap set on $\intervalset[a,b]$ that includes $I$ in the $k$-overlap set (if one exists).
		When computing such a value $\dms{k}(I)$, we consider all subsets $\mathcal{J} \subseteq \mathcal{P}(I,\intervalset)$ of cardinality $|\mathcal{J}| \le k$ of at most $k$ neighbors of $I$ to be included in a $k$-overlap set, while $\mathcal{P}(I,\intervalset) \setminus \mathcal{J}$ is excluded.

		For keeping track of how many intervals are still allowed to overlap each interval we introduce the \emph{capacity} of each interval boundary $ i \in \sigma(\intervalset) = \{1, 2, \dots, 2n\}$. These capacities are stored in a vector $ \lambda = (\lambda_1,\dots,\lambda_{2n}) $, where each $\lambda_i$ is the capacity of the interval boundary $i \in \sigma(\intervalset)$. 
		Each $\lambda_i$ is basically a value in the set $\{0,1,\dots, k\}$ that indicates how many additional intervals may still overlap the interval corresponding to $i$, see Fig.~\ref{fig:gen_dms}.
		We actually define $ \dms{k}_\lambda([a,b]) $ as the maximum weight of a $k$-overlap set in $ \intervalset[a,b] $ with pre-defined capacities $\lambda$. 
		In the appendix we prove that the number of relevant vectors $\lambda$ to consider for each interval can be bounded by $O(\gamma^k)$, where $\gamma$ is the maximum degree of the overlap graph corresponding to $\intervalset$.

	\begin{figure}[tb]
		\centering
		\includegraphics{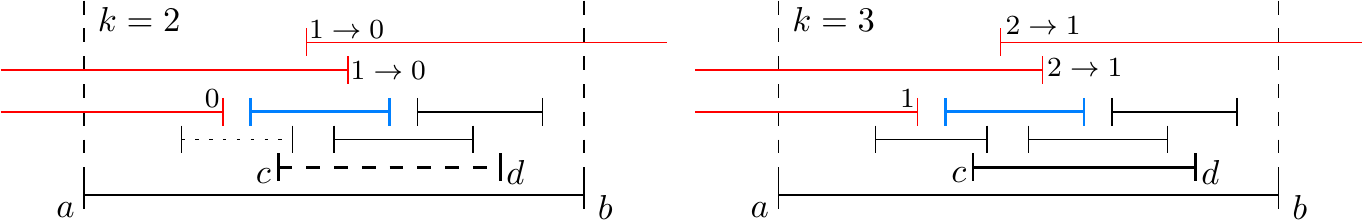}
		\centering
		\caption{Examples for $ k = 2 $ and $ k = 3 $. The red intervals are in a solution set. The arrows indicate how the capacities change if the blue interval is included in a solution. For $ k = 2 $ we cannot use the interval $ [c,d] $ since some capacities are zero, but for $ k = 3 $ it remains possible. }
		\label{fig:gen_dms} 
	\end{figure}
	
		For our recursive definition we assume that when computing $\dms{k}_\lambda(\iI)$ all values 	$ \dms{k}_\lambda(\iJ) $ with $ \iJ \in \intervalset $ and $ \length(\iJ) < \length(\iI) $ are already computed. The following recurrence computes one $ \dms{k}_\lambda(\iI) $ value given a valid capacity vector $ \lambda $ and an interval $ \iI = [a,b] \in \intervalset $
		\begin{align}
			\label{rec:gen_dms}
			\dms{k}_\lambda([a,b]) = S_{I,\lambda}[a+1] + w([a,b]).
		\end{align}
		This means that we select $I$ for the $k$-overlap set, add its weight $w(I)$, and recursively solve the subinstance of intervals nested in $I$ subject to the capacities $\lambda$.
		As in the approach for the $ \dms{1} $ values the main work is done in recurrence $ S_{I,\lambda}[x] $ where $ x \in \sigma(\intervalset[a,b]) $. 
		In the appendix we prove Lemma~\ref{lem:fix_lambdaApx}, which shows the correctness of this computation using a similar induction-based proof as Lemma~\ref{lem:dmsI} for $ k = 1 $, but being more careful with the computation of the correct weights.
		Lemma~\ref{lem:fix_lambda} is a simplified version of Lemma~\ref{lem:fix_lambdaApx}.
		\begin{lemma}
			\label{lem:fix_lambda}
			Let \intervalset be a set of intervals, $\iI \in \intervalset$, $\lambda$ a valid capacity vector for $\iI$, and $\gamma$ the maximum degree of the corresponding overlap graph. 
			Then $ \dms{k}_\lambda(\iI) $ can be computed in $ O(\gamma^k\length(\iI))$ time once the $ \dms{k}_{\lambda}(\iJ) $ values are computed for all $J\in \intervalset$ with $ \length(\iJ) < \length(\iI) $. %
		\end{lemma}
		Applying Lemma~\ref{lem:fix_lambda} to all $ \iI \in \intervalset $ and all valid capacity vectors $ \lambda $ results in a running time of $ O(\gamma^{2k}\ell) $ to compute all $ \dms{k}_\lambda(\iI) $ values, where $\ell$ is the total interval length (see Lemma~\ref{lem:gen_computeall} in Appendix~\ref{apx:XP}).

		Now that we know how to compute all values $\dms{k}_\lambda(I)$ for all $I \in \intervalset$ and all relevant capacity vectors $\lambda$, we can obtain the optimal solution by introducing a dummy interval $\hat{I}$ with weight $w(\hat{I})=0$ that nests the entire set $\intervalset$. We compute the value $\dms{k}_{\hat{\lambda}}(\hat{I})$ for a capacity vector $\hat{\lambda}$ that puts no prior restrictions on the intervals in $\intervalset$. This solution obviously contains the max-weight $k$-overlap set for $\intervalset$. We summarize:

		\begin{restatable}{theorem}{thmgeneral}\label{thm:general}
			A max-weight $k$-overlap set for a set of intervals \intervalset  %
			can be computed in $ O(\gamma^{2k}\ell) \subseteq O(|\intervalset|^{2k+2})$ time, where $\ell$ is the total interval length and $\gamma$ is the maximum degree of the corresponding overlap graph.
		\end{restatable}
	
		The running time in Theorem~\ref{thm:general} implies that both the max-weight $k$-overlap set problem and the equivalent \ksubgraphN{k} problem for circle (overlap) graphs are in \XP.\footnote{The class \XP\ contains problems that can be solved in time $O(n^{f(k)})$, where $n$ is the input size, $k$ is a parameter, and $f$ is a computable function.} This fact alone can alternatively be derived from a metatheorem of Fomin et al.~\cite{fomin2015large} as follows.\footnote{We thank an anonymous reviewer of an earlier version for pointing us to this fact.} 
		The number of minimal separators of circle graphs can be polynomially bounded by $O(n^2)$ as shown by Kloks~\cite{kloks1996treewidth}. Further, since we are interested in a bounded-degree induced subgraph $G[V']$ of a circle graph $G$, we know from Gaspers et al.~\cite{gaspers2009exponential} that $G[V']$ has treewidth at most four times the maximum degree $k$. With these two pre-conditions the metatheorem of Fomin et al.~\cite{fomin2015large} yields the existence of an \XP-time algorithm for \ksubgraphN{k} on circle graphs.
However, the running time obtained from Fomin et al.~\cite{fomin2015large} is $ O(|\Pi_G|\cdot n^{t+4} \cdot f(t,\phi)) $ where $ |\Pi_G| $ is the number of potential cliques in $ G $, $ t $ is the treewidth of $ G[V'] $ with $ V' \subseteq V $ being the solution set, and $ f $ is a tower function depending only on $ t $ and the CMSO (Counting Monadic Second Order Logic) formula $ \phi $ (compare Thomas~\cite{thomas1996languages} proving this already for MSO formulas). Let $ k $ be the desired degree of a $ \ksubgraphN{k} $ instance, then the treewidth of $G[V']$ is at most $ 4k $. Further by Kloks~\cite{kloks1996treewidth} we know $ |\Pi_G| = O(n^2) $. Hence the running-time of the algorithm would be in $ O(n^{4k+6} \cdot f(4k, \phi)) $, whereas our problem-specific algorithm has running time $ O(n^{2k+2})$. %

		\section{Experiments}\label{sec:experiments}
		We implemented the algorithm for \ksubgraphN{1} from Section~\ref{sec:algo} and the independent set algorithm for \ksubgraphN{0} in \CC. %
		The compiler was \GG, version $ 7.2.0 $ with set -O$ 3 $ flag. Further we used the OGDF library~\cite{chimani2013open} in its most current snapshot version. Experiments were run on a standard desktop computer with an eight core Intel i7-6700 CPU clocked at $ 3.4 $ GHz and $ 16 $ GB RAM, running Archlinux and kernel version $ 4.13.12 $. The implementation is available under \url{https://www.ac.tuwien.ac.at/two-sided-layouts/}.
		
		We generated two sets of random biconnected graphs using OGDF. The first set has $ 5,156 $ and the second $ 4,822 $ different non-planar graphs. We varied the edge to vertex ratio between $ 1.0 $ and $ 5.0 $ and the number of vertices between 20 and 60.  In addition, we used the Rome graph library (\url{http://www.graphdrawing.org/data.html}) consisting of $ 8,504 $ non-planar graphs with a density of $ 0.5 $ to $ 2.1 $ and 10 to 100 vertices. The relative sparsity of the test instances is not a drawback, since it is impractical to use circular layouts for visualizing very dense graphs. For dense graphs one would have to apply some form of bundling strategy to reduce edge clutter. Given such a bundled layout one could consider minimizing bundled crossings~\cite{alam2016bundled} and adapt our algorithms to the resulting ``bundled'' circle graph.

		For the random test graphs Fig.~\ref{fig:saved_crossings} displays the percentage of crossings saved by the layouts with exterior edges versus the one-sided circular layout implemented in OGDF. Compared to the approach without exterior crossings we find that allowing up to one crossing per edge in the exterior one can save around 11\% more crossings on average for the random instances and $ 7.5 $\%  for the Rome graphs. Setting the edge weight in \ksubgraphN{1} to one (i.e., counting exterior crossings in the optimization) or two (i.e., not counting exterior crossings in the optimization), has (almost) no noticeable effect. %
		
		Figure~\ref{fig:computation_time} depicts the times needed to compute the layouts for the respective densities. We observe the expected behaviour. The case of $ k = 0 $  with $O(|\mathcal{E}|^2)$ time is a lot faster as the graphs get more dense. Still for our sparse test instances our algorithm for \ksubgraphN{1} with $ O(|\mathcal{E}|^4) $ time runs sufficiently fast to be used on graphs with up to $ 60 $ vertices ($ 82 $ seconds on average). For additional plots we refer to Appendix~\ref{apx:plots}.
				
		\begin{figure}
			\centering
			\begin{subfigure}[t]{0.48\textwidth}
				\input{plots/random3/ag_saved_crossings.tex}
				\subcaption{Density vs the percentage of saved crossings.}
				\label{fig:saved_crossings}
			\end{subfigure}
			\hfill
			\begin{subfigure}[t]{0.48\textwidth}
				\input{plots/random3/ag_time.tex}
				\subcaption{Density vs the computation time.}				
				\label{fig:computation_time}				
			\end{subfigure}
			\begin{subfigure}[t]{0.5\textwidth}
				\vspace{-5.2cm}
				\hspace{-0.55cm}
				\input{plots/legend}
			\end{subfigure}

			\caption{Plots for the test set with 4881 graphs. $ w $ is the weight given to the edges of the circle graph, see Section~\ref{sub:problem_transformation}. This test set has between 20 and 60 vertices and an average density of $ 2.6 $.}
			\label{fig:saved}		
		\end{figure}
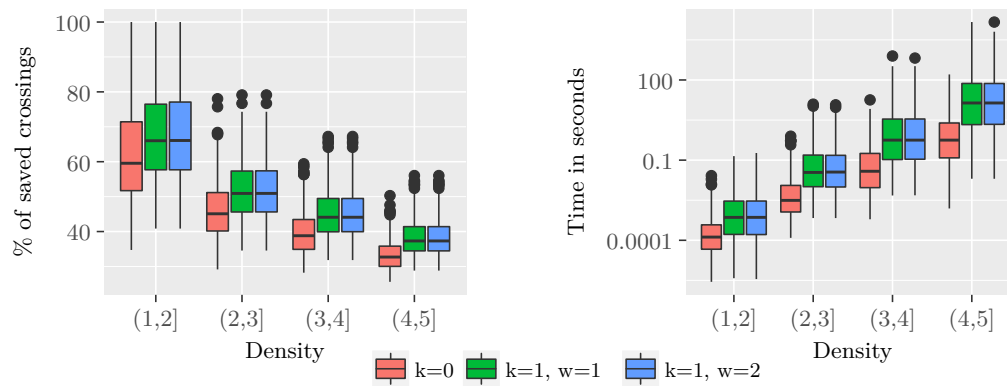

		Our tests show a clear improvement in crossing reduction when going from $ k = 0 $ to $ k = 1 $. Of course this comes with a non-negligible runtime increase. For the practically interesting sparse instances, though, \ksubgraph{1} can be solved fast enough to be useful in practice.
		
		\section{Open questions}		

			The overall hardness of the \ksubgraphN{k} problem on circle graphs, parametrized by just the desired degree $ k $ remains open. While we could show 
			\NP-hardness, we do not know whether an \FPT-algorithm exists or whether the problem is $\W[1]$-hard.
			In terms of the motivating graph layout problem crossing minimization is known as a major factor for readability. 
			Yet, practical two-sided layout algorithms must also apply suitable vertex-ordering heuristics and they should further take into account the length and actual routing of exterior edges.
			Edge bundling approaches for both interior and exterior edges promise to further reduce visual clutter, but then bundling and bundled crossing minimization should be considered simultaneously.
			It would also be interesting to generalize the problem from circular layouts to other layout types, where many but not necessarily all vertices can be fixed on a boundary curve.
		\bibliographystyle{plainurl}
		\bibliography{p53-klute}
		
		\newpage
		\appendix

		\section{Omitted proofs from Section~\ref{sec:general}}\label{apx:XP}
		
		We start by describing the capacity vectors in more detail.
		Each capacity $\lambda_i$ is a value in the set $  \{0,1,\dots,k,\infty, \bot\} $. If $ \lambda_i = \bot $ we call it \emph{undefined}, i.e., no decision about its interval has been made. If $ \lambda_i = \infty $ we say it has \emph{unlimited} capacity, which means that the corresponding interval is not selected for the candidate solution. Finally if $ \lambda_i = \alpha \in \{0,1,\dots,k\} $ we say it has capacity $ \alpha $, meaning that we can still add $ \alpha $ more intervals which overlap the interval corresponding to $ i $, recall Fig.~\ref{fig:gen_dms}.

			For an interval $[a,b]$ we use the short-hand notation $ \lambda_{a,b} = x $ to state that $ \lambda_a = x $ and $ \lambda_b = x $; likewise we use $ \lambda_{a,b} \neq x $ to state that $ \lambda_a \neq x $ and $ \lambda_b \neq x $.
		
			We say a capacity vector $ \lambda $ is \emph{valid for interval $ \iI = [a,b]$} if both capacities $ \lambda_{a,b} \not\in \{\bot, \infty \} $ and at most $ k $ intervals $ [c,d] \in \mathcal{P}(\iI,\intervalset) $ have $ \lambda_{c,d} \neq \infty $. 
			We say $ \lambda $ is \emph{valid for \intervalset} if it is valid for every $ \iI \in \intervalset $.

			\begin{definition}
				Let \intervalset be a set of intervals, $ \iI = [a,b] \in \intervalset $  and $ \lambda $ a valid capacity vector for $ \intervalset $ such that $ \lambda_{a,b} \neq \infty $. Then $ \lambda' $ is a \emph{legal successor} of $ \lambda $ for $\iI$ (written $ \lambda' \leftarrow_\iI \lambda $) if there is a set $\mathcal{J} \subseteq \mathcal{P}(\iI,\intervalset)$ with $|\mathcal{J}|\le k$ such that for all $[i,j] \in \mathcal{J}$ we have $\lambda_{i,j} \ne \infty$, for all $[i,j] \in \mathcal{P}(\iI,\intervalset) \setminus \mathcal{J}$ we have $\lambda_{i,j} \in \{\bot,\infty\}$ and such that for all $ [x,y] \in \intervalset $ 
				\begin{align*}
				\lambda_{x}'  &= \lambda_{x}, \quad  \lambda_{y}' = \lambda_{y} &  &\text{if }  [x,y] \not\in \mathcal{P}(\iI,\intervalset) \cup \{\iI\}  \\				
				\lambda'_{x,y} &= 0 && \text{if } [x,y] = [a,b] = \iI\\
				\lambda'_{x,y} &= \infty && \text{if } [x,y] \in \mathcal{P}(\iI,\intervalset) \setminus \mathcal{J}\\
				\lambda_x' &= \alpha, \quad \lambda_y' = k - t_{x,y} - \alpha \quad \text{for } \alpha \in [0,k-t_{x,y}]  && \text{if } [x,y] \in \mathcal{J} \text{ and } \lambda_{x,y} = \bot  \\
				\lambda_x' &= \lambda_x - l_{x,y}, \quad \lambda_y' = \lambda_y - r_{x,y} && \text{if } [x,y] \in \mathcal{J} \text{ and } \lambda_{x,y} \ne \bot,
				\end{align*}
				where $ t_{x,y} = |\{ [i,j] \mid [i,j] \in \mathcal{P}([x,y],\intervalset) \text{ and } \left( [i,j] \in \mathcal{J} \text{ or } \lambda_{i,j} \not\in \{\infty,\bot\} \right)\}| $,  $ l_{x,y} = |\{ [i,j] \mid [i,j] \in \mathcal{J} \cup \{I\}, i < x < j, \lambda_{i,j} = \bot \}|$, and $ r_{x,y} = | \{ [i,j] \mid [i,j] \in \mathcal{J} \cup \{I\}, i < y < j, \lambda_{i,j} = \bot \} |$.
			\end{definition}
			
			Intuitively, a legal successor of a valid capacity vector for $I$ fixes the interval $I$ to be in the considered $k$-overlap set and determines for all intervals $J \in \mathcal{P}(I,\intervalset)$ whether they are selected for the $k$-overlap set ($J \in \mathcal J$) or not ($J \not\in \mathcal J$). All affected capacities are updated accordingly. Capacities can only be decreased and if an interval $[i,j]$ has been discarded previously in $\lambda$, i.e., $\lambda_{i,j} = \infty$, then it cannot be contained in $\mathcal J$.

			\begin{definition}
				Let \intervalset be a set of intervals, $ \iI \in \intervalset $  and $ \lambda $ a valid capacity vector for $ \iI $, then 
							\begin{align*}
				L'(\lambda,\iI,\intervalset) = \{ \lambda' \mid \lambda' \leftarrow_\iI \lambda \}
							\end{align*}
				is the set of all legal successors of $\lambda$ for $\iI$.

			\end{definition}

					We now assume that all $ \dms{k}_\lambda(\iJ) $ with $ \iJ \in \intervalset $ and $ \length(\iJ) < \length(\iI) $ are already computed. The following recurrence computes one $ \dms{k}_\lambda(\iI) $ value given a valid capacity vector $ \lambda $ and an interval $ \iI = [a,b] \in \intervalset $
			\begin{align}
			\label{rec:gen_dmsApx}
			\dms{k}_\lambda([a,b]) = S_{I,\lambda}[a+1] + w(\iI).
			\end{align}
			
		It remains to describe the computation of the recurrence for $S_{I,\lambda}$. 
		To consider the correct weight for the optimal solution we have to pay attention whenever a capacity is changed from $ \bot $ to some $ \alpha \in [0,k] $. We introduce the following notation for the set of all intervals  changed between capacity vectors $ \lambda' $ and $ \lambda $
\begin{align*}
\text{new}(\lambda',\lambda,\iK,\intervalset') &= \{ [i,j] \mid [i,j] \in {\mathcal{P}}(\iK,\intervalset') \land \lambda_{i,j} = \bot \land 0 \le \lambda_{i,j}' \le k \}
\end{align*}
and for the set of all intervals overlapping a given interval, but not in the set $ \text{new}(\lambda',\lambda,\iK,\intervalset') $
\begin{align*}
\text{old}(\lambda', \lambda,\iK,\intervalset') &= \{ [i,j] \mid [i,j] \in \mathcal{P}(\iK,\intervalset')\setminus\text{new}(\lambda',\lambda,\iK,\intervalset')) \land 0 \le \lambda_{i,j} \le k \}.
\end{align*}
Secondly we introduce the weight function $ w(\lambda',\lambda,\iK,\intervalset') $, which for two capacity vectors $ \lambda' $ and $ \lambda $ and one interval $ \iK $ computes the weight of all intervals newly considered in $ \lambda' $, minus the weight of all overlaps (edges) involving new intervals
\begin{align*}
w(\lambda',\lambda,\iK,\intervalset') &= \sum_{L \in \text{new}(\lambda',\lambda,\iK,\intervalset')} w(L) -
\Bigg(\sum_{L \in \text{new}(\lambda',\lambda,\iK,\intervalset')}\\
&\Bigg(w(\iK,L) + \frac{1}{2}\sum_{\substack{M\in \text{new}(\lambda',\lambda,\iK,\intervalset')\\\cap \mathcal{P}(L,\intervalset')}} w(L,M) +
\sum_{\substack{M \in \text{old}(\lambda',\lambda,\iK,\intervalset')\\\cap \mathcal{P}(L,\intervalset')}} 
w(L,M)\Bigg)\Bigg).
\end{align*}

		Recall Equation~(\ref{rec:gen_dms}) for $\dms{k}_\lambda([a,b])$.
		As in the approach for the $ \dms{1} $ values the main work is done in recurrence $ S_{I,\lambda}[x] $ where $ x \in \sigma(\intervalset[a,b]) $. 
We set $ S_{I,\lambda}[b] = 0 $ and for any right endpoint $ d $ of an interval $ \iJ = [c,d] \in \intervalset[a,b] $ we set $ S_{I,\lambda}[d] = S_{I,\lambda}[d + 1] $. For each left endpoint $ c $ we use the following maximization
\[
S_{I,\lambda}[c] = \max\{ \{S_{I,\lambda''}[c+1]\} \cup \{H_{I,\lambda'}([c,d]) \mid \lambda' \in L'(\lambda,[c,d], \intervalset %
\} \},
\]
where $\lambda'' \leftarrow_I \lambda $ with $\lambda_{c,d}'' = \infty$ and		
$ H_{I,\lambda'}([c,d]) = \dms{k}_{\lambda'}([c,d]) + S_{I,\lambda'}[d + 1] + w(\lambda',\lambda,\iI,\intervalset[a,b]) $.

\begin{lemma}\label{lem:successor_size}
	Let \intervalset be a set of intervals, $\iI \in \intervalset$, and $\gamma$ the maximum degree of the corresponding overlap graph. 
	\begin{enumerate}
		\item There are $O(\gamma^k (k+1)^k)$ valid basic capacity vectors for $\iI$.
		\item For each valid capacity vector $\lambda$ the set $L'(\lambda,\iI,\intervalset)$ of legal successors has $O(\gamma^k (k+1)^k)$ elements.
	\end{enumerate}
\end{lemma}
\begin{proof}
	Since the set $\mathcal{P}(I,\intervalset)$ has at most $\gamma$ elements, there are at most $\gamma$ positions in any capacity vector that influence the options for $\iI$. We can ignore the capacities of positions that do not belong to intervals in $\mathcal{P}(I,\intervalset) \cup \{I\}$ and consider all capacity vectors to belong to the same class as long as they coincide on the capacities of positions belonging to $\mathcal{P}(I,\intervalset) \cup \{I\}$. We call each such class a \emph{basic} capacity vector for $\iI$. To be valid for $\iI$ at most $k$ of those positions in $\lambda$ may have a value in $\{0,1,\dots, k\}$. There are $O(\binom{\gamma}{k}) = O(\gamma^k)$ different combinations of up to $k$ positions, and in each position up to $k+1$ different values, which yields the first bound.
	
	For a valid capacity vector $\lambda$, a legal successor is determined by considering all $O(\gamma^k)$ choices for a set $\mathcal{J}$ of at most $k$ intervals, and for each chosen interval there are up to $k+1$ ways of splitting the capacities between left and right index. This gives again a bound of $O(\gamma^k (k+1)^k)$.
\end{proof}

In the following we define $ z = \gamma^{k}(k+1)^{k} $ as a shorthand. The next lemma is a more formal version of Lemma~\ref{lem:fix_lambda} stated in Section~\ref{sec:general}.

			\begin{lemma}
				\label{lem:fix_lambdaApx}
						Let \intervalset be a set of intervals, $\iI \in \intervalset$, and $\lambda$ a valid capacity vector for $\iI$ then the value $ \dms{k}_\lambda(\iI) $ can be computed in $ O(z\length(\iI))$ time once the $ \dms{k}_{\lambda'}(\iJ) $ values are computed for all $ \iJ \in \intervalset $ with $ \length(\iJ) < \length(\iI) $ and $ \lambda' $ a valid basic capacity vector for $\iJ$.
			\end{lemma}
				\begin{proof}			
			The proof works similar to the proof of Lemma~\ref{lem:dmsI} for $\dms{1}$ by induction over the size of $ \mathcal{N}(\iI,\intervalset) $. Let $\iI = [a,b]$ and suppose $ \mathcal{N}(\iI,\intervalset) = \emptyset $, then $ \dms{k}_\lambda(\iI) = S_{I,\lambda}[a+1] + w(\iI) $, and, regardless of the capacity vector $ \lambda $,  $ S_{I,\lambda}[a+1] = S_{I,\lambda}[b] = 0 $ by definition.

			Now consider $ \mathcal{N}(\iI,\intervalset) $ with at least one interval. For every left endpoint $ i \in \sigma(\mathcal{N}(\iI,\intervalset)) $ we have the choice of using the interval $ \iJ = [i,j] \in \intervalset $ starting at $ i $ or discarding it. 
			
			When we decide not to use $ \iJ $ the recurrence correctly continues with the next endpoint $i+1$ and sets the capacities of $ i $ and $ j $ to $ \infty $. In case we choose to use $ \iJ $ we have to consider which intervals from $ {\mathcal{P}}(\iJ,\intervalset) $  we additionally take into consideration for the optimal solution. For this we exhaustively consider all possible combinations and capacities of intervals in $ {\mathcal{P}}(\iJ,\intervalset) $. These are exactly the legal successors $ \lambda' \in L'(\lambda,\iJ,\intervalset) $ over which we maximize.
			
			It remains to argue that the weight $ \dms{k}_\lambda(\iI) $ is correct in the end. First consider the weight $ w(\iJ) $ of the intervals $ \iJ =[i,j] \in \intervalset[a,b] $. There are two possibilities how $ w(\iJ) $ can be added. If it is picked as the next interval in a $ \dms{k}_\lambda(\iJ) $ call then $w(\iJ)$ is added in Recurrence~(\ref{rec:gen_dmsApx}). The other possibility is for $ \lambda_{i}' $ and $\lambda_j'$ to be set to some values $ \alpha$ and $ k-t_{x,y} - \alpha $ in a legal successor. In that case we add the weight $w(\iJ)$ in the $ w(\lambda',\lambda,\iJ,\intervalset[a,b]) $ term. If we decide not to consider $ \iJ $ its weight is not included either.
			
			We further have to subtract the weight of all intersecting pairs $ \iJ = [i,j], \iK = [x,y] \in \intervalset[a,b] $ which are both included in the solution set. Without loss of generality this is done in the weight term $ w(\lambda',\lambda,\iK,\intervalset[a,b]) $. There are three cases to consider: (i) we are in a $ \dms{k}_\lambda(\iJ) $ call, (ii) $ \lambda_{i,j} $ was set  some steps before $\lambda_{x,y}$, or (iii) $ \lambda_{i,j} $ is set in the same step as $ \lambda_{x,y} $.
			
			Let's look at the two latter cases first. There must be an interval $ L \in \intervalset[a,b] $ such that we are inside a call to $ \dms{k}_\lambda(L) $. That means we consider the weight term $ w(\lambda', \lambda, L, \intervalset[a,b]) $. We know $ \iK \in \text{new}(\lambda',\lambda,L,\intervalset[a,b]) $. The two cases for $ \iJ $ are $ \iJ \in \text{old}(\lambda',\lambda,L,\intervalset[a,b]) $ and $ \iJ \in \text{new}(\lambda',\lambda,L,\intervalset[a,b]) $. Both are considered in the term $ w(\lambda', \lambda, L, \intervalset[a,b]) $ and we correctly divide by two to correct the double-counting for $ \iJ \in \text{new}(\lambda',\lambda,L,\intervalset[a,b]) $.
			
			In case (i) we are inside a call to $ \dms{k}_\lambda(\iJ) $, and still $ \iK \in \text{new}(\lambda',\lambda,\iJ,\intervalset[a,b]) $. Here we use $ w(\lambda', \lambda, \iJ, \intervalset[a,b]) $, where the weight for newly added edges is subtracted correctly.
			
			The running time of $ O(z\length(\iI)) $ follows from Lemma~\ref{lem:successor_size} (2).
		\end{proof}
	\begin{lemma}
		\label{lem:var_lambdaApx}
		Let \intervalset be a set of intervals and $\iI \in \intervalset$, then the values $ \dms{k}_\lambda(\iI) $ can be computed for all valid basic capacity vectors $ \lambda $ in $ O(z^2\length(\iI)) $ time once the values $ \dms{k}_{\lambda'}(\iJ) $ are computed for all $ \iJ \in \intervalset $ with $ \length(\iJ) < \length(\iI) $ and for all valid basic capacity vectors $ \lambda' $.
	\end{lemma}
	\begin{proof}
		We apply Lemma~\ref{lem:fix_lambdaApx} for each valid basic capacity vector of $\iI$. By Lemma~\ref{lem:successor_size} (1) there are $O(z)$ such capacity vectors, which implies the time bound of $ O(z^2\length(\iI)) $.
	\end{proof}
	\begin{lemma}
		\label{lem:gen_computeall}
		Let \intervalset be a set of intervals, then the values $ \dms{k}_\lambda(\iI) $ can be computed for all $ \iI \in \intervalset $ and for all valid basic capacity vectors $\lambda$ in $ O(z^2\ell) $ time.
	\end{lemma}
	\begin{proof}
		We apply Lemma~\ref{lem:var_lambdaApx} for every $\iI \in \intervalset$ and obtain the running time of $ O(z^2\ell) $ since $ \sum_{\iI \in \intervalset}\length(\iI) = \ell $.
	\end{proof}
		\thmgeneral*

\begin{proof}
	For computing the maximum weight of a $k$-overlap set of $\intervalset$ we can introduce a dummy interval $\hat{I}=[0,2n+1]$ with weight $w(\hat{I})=0$ that nests the entire set $\intervalset$, i.e., $\mathcal{N}(\hat{I},\intervalset \cup \{\hat{I}\}) = \intervalset$. We define the initial capacity vector $\hat{\lambda}$ that has $\hat{\lambda}_{0,2n+1} = 0$ and all other $\hat{\lambda}_i = \bot$. Using the computation scheme of Lemma~\ref{lem:gen_computeall} we obtain in $O(z^2 \ell) = O((k+1)^{2k} \gamma^{2k} \ell) = O(\gamma^{2k} \ell)$ time all values $\dms{k}_\lambda(\iI)$, including $\dms{k}_{\hat{\lambda}}(\hat{I})$. Note that for this time bound we use that $k$ is an arbitrary but fixed integer and thus $O((k+1)^k) = O(1)$. Clearly the solution corresponding to the value $\dms{k}_{\hat{\lambda}}(\hat{I})$ includes a max-weight $k$-overlap set for $\mathcal{N}(\hat{I},\intervalset \cup \{\hat{I}\}) = \intervalset$.
\end{proof}

\section{Additional plots and figures}\label{apx:plots}

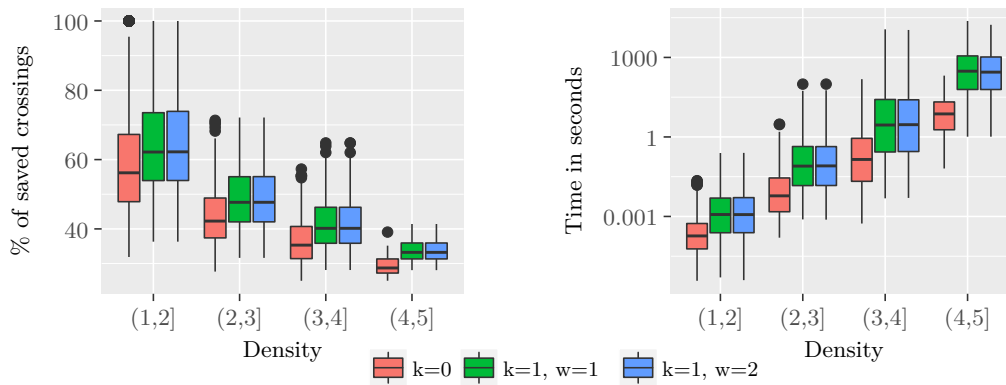
\begin{figure}[!h]

			\centering
\begin{subfigure}[t]{0.48\textwidth}
	\input{plots/random4/ag_saved_crossings.tex}
	\subcaption{Density vs the percentage of saved crossings.}
	\label{fig:saved_crossings:random3}
\end{subfigure}
\hfill
\begin{subfigure}[t]{0.48\textwidth}
	\input{plots/random4/ag_time.tex}
	\subcaption{Density vs the computation time.}				
	\label{fig:computation_time:random3}				
\end{subfigure}
\begin{subfigure}[t]{0.5\textwidth}
	\vspace{-5.2cm}
	\hspace{-0.55cm}
	\input{plots/legend}
\end{subfigure}

\caption{Plots for the test set with 5156 graphs. $ w $ is the weight given to the edges of the circle graph, see Section~\ref{sub:problem_transformation}. This set is larger, but the number of vertices is lower (20 to 50), while the average density is a little bit higher (3.0 compared to 2.6). Nonetheless one observes the same behaviour as for the other test instance.}
\label{fig:saved:random4}		
\end{figure}

\begin{figure}[!h]
			\centering
\begin{subfigure}[t]{\textwidth}
	\input{plots/rome/ag_saved_crossings.tex}
	\subcaption{Density vs the percentage of saved crossings.}
	\label{fig:saved_crossings:rome}
\end{subfigure}

\begin{subfigure}[t]{\textwidth}
	\input{plots/rome/ag_time.tex}
	\subcaption{Density vs the computation time.}				
	\label{fig:computation_time:rome}				
\end{subfigure}
\begin{subfigure}[t]{0.5\textwidth}
	\vspace{-4.8cm}
	\hspace{2.5cm}
	\input{plots/legend}
\end{subfigure}

\caption{Plots for the Rome graphs. $ w $ is the weight given to the edges of the circle graph, see Section~\ref{sub:problem_transformation}. The parabola like behaviour of the means in Fig.~\ref{fig:saved_crossings:rome} and~\ref{fig:computation_time:rome} seems to be an effect of the structure of the Rome graphs since our randomly generated graphs do not exhibit the same behaviour.}
\label{fig:saved:rome}		
\end{figure}

\end{document}

%% file: plots/random3/ag_saved_crossings.tex
\begin{tikzpicture}[x=1pt,y=1pt]
\definecolor{fillColor}{RGB}{255,255,255}
\path[use as bounding box,fill=fillColor,fill opacity=0.00] (0,0) rectangle (180.67,144.54);
\begin{scope}
\path[clip] (  0.00,  0.00) rectangle (180.67,144.54);
\definecolor{drawColor}{RGB}{255,255,255}
\definecolor{fillColor}{RGB}{255,255,255}

\path[draw=drawColor,line width= 0.6pt,line join=round,line cap=round,fill=fillColor] ( -0.00,  0.00) rectangle (180.67,144.54);
\end{scope}
\begin{scope}
\path[clip] ( 39.63, 31.23) rectangle (175.17,139.04);
\definecolor{fillColor}{gray}{0.92}

\path[fill=fillColor] ( 39.63, 31.23) rectangle (175.18,139.04);
\definecolor{drawColor}{RGB}{255,255,255}

\path[draw=drawColor,line width= 0.3pt,line join=round] ( 39.63, 41.92) --
	(175.17, 41.92);

\path[draw=drawColor,line width= 0.3pt,line join=round] ( 39.63, 68.27) --
	(175.17, 68.27);

\path[draw=drawColor,line width= 0.3pt,line join=round] ( 39.63, 94.62) --
	(175.17, 94.62);

\path[draw=drawColor,line width= 0.3pt,line join=round] ( 39.63,120.97) --
	(175.17,120.97);

\path[draw=drawColor,line width= 0.6pt,line join=round] ( 39.63, 55.10) --
	(175.17, 55.10);

\path[draw=drawColor,line width= 0.6pt,line join=round] ( 39.63, 81.45) --
	(175.17, 81.45);

\path[draw=drawColor,line width= 0.6pt,line join=round] ( 39.63,107.79) --
	(175.17,107.79);

\path[draw=drawColor,line width= 0.6pt,line join=round] ( 39.63,134.14) --
	(175.17,134.14);

\path[draw=drawColor,line width= 0.6pt,line join=round] ( 59.00, 31.23) --
	( 59.00,139.04);

\path[draw=drawColor,line width= 0.6pt,line join=round] ( 91.27, 31.23) --
	( 91.27,139.04);

\path[draw=drawColor,line width= 0.6pt,line join=round] (123.54, 31.23) --
	(123.54,139.04);

\path[draw=drawColor,line width= 0.6pt,line join=round] (155.81, 31.23) --
	(155.81,139.04);
\definecolor{drawColor}{gray}{0.20}

\path[draw=drawColor,line width= 0.6pt,line join=round] ( 49.85, 96.50) -- ( 49.85,134.14);

\path[draw=drawColor,line width= 0.6pt,line join=round] ( 49.85, 70.53) -- ( 49.85, 48.15);
\definecolor{fillColor}{RGB}{248,118,109}

\path[draw=drawColor,line width= 0.6pt,line join=round,line cap=round,fill=fillColor] ( 45.82, 96.50) --
	( 45.82, 70.53) --
	( 53.89, 70.53) --
	( 53.89, 96.50) --
	( 45.82, 96.50) --
	cycle;

\path[draw=drawColor,line width= 1.1pt,line join=round] ( 45.82, 80.86) -- ( 53.89, 80.86);
\definecolor{fillColor}{gray}{0.20}

\path[draw=drawColor,line width= 0.4pt,line join=round,line cap=round,fill=fillColor] ( 82.12, 92.16) circle (  1.96);

\path[draw=drawColor,line width= 0.4pt,line join=round,line cap=round,fill=fillColor] ( 82.12, 92.47) circle (  1.96);

\path[draw=drawColor,line width= 0.4pt,line join=round,line cap=round,fill=fillColor] ( 82.12,102.16) circle (  1.96);

\path[draw=drawColor,line width= 0.4pt,line join=round,line cap=round,fill=fillColor] ( 82.12,105.19) circle (  1.96);

\path[draw=drawColor,line width= 0.4pt,line join=round,line cap=round,fill=fillColor] ( 82.12, 91.64) circle (  1.96);

\path[draw=drawColor,line width= 0.6pt,line join=round] ( 82.12, 69.80) -- ( 82.12, 90.58);

\path[draw=drawColor,line width= 0.6pt,line join=round] ( 82.12, 55.31) -- ( 82.12, 40.82);
\definecolor{fillColor}{RGB}{248,118,109}

\path[draw=drawColor,line width= 0.6pt,line join=round,line cap=round,fill=fillColor] ( 78.09, 69.80) --
	( 78.09, 55.31) --
	( 86.16, 55.31) --
	( 86.16, 69.80) --
	( 78.09, 69.80) --
	cycle;

\path[draw=drawColor,line width= 1.1pt,line join=round] ( 78.09, 61.80) -- ( 86.16, 61.80);
\definecolor{fillColor}{gray}{0.20}

\path[draw=drawColor,line width= 0.4pt,line join=round,line cap=round,fill=fillColor] (114.40, 79.96) circle (  1.96);

\path[draw=drawColor,line width= 0.4pt,line join=round,line cap=round,fill=fillColor] (114.40, 77.56) circle (  1.96);

\path[draw=drawColor,line width= 0.4pt,line join=round,line cap=round,fill=fillColor] (114.40, 76.83) circle (  1.96);

\path[draw=drawColor,line width= 0.4pt,line join=round,line cap=round,fill=fillColor] (114.40, 77.68) circle (  1.96);

\path[draw=drawColor,line width= 0.4pt,line join=round,line cap=round,fill=fillColor] (114.40, 76.57) circle (  1.96);

\path[draw=drawColor,line width= 0.4pt,line join=round,line cap=round,fill=fillColor] (114.40, 76.69) circle (  1.96);

\path[draw=drawColor,line width= 0.4pt,line join=round,line cap=round,fill=fillColor] (114.40, 80.66) circle (  1.96);

\path[draw=drawColor,line width= 0.4pt,line join=round,line cap=round,fill=fillColor] (114.40, 79.34) circle (  1.96);

\path[draw=drawColor,line width= 0.6pt,line join=round] (114.40, 59.64) -- (114.40, 76.39);

\path[draw=drawColor,line width= 0.6pt,line join=round] (114.40, 48.38) -- (114.40, 39.58);
\definecolor{fillColor}{RGB}{248,118,109}

\path[draw=drawColor,line width= 0.6pt,line join=round,line cap=round,fill=fillColor] (110.36, 59.64) --
	(110.36, 48.38) --
	(118.43, 48.38) --
	(118.43, 59.64) --
	(110.36, 59.64) --
	cycle;

\path[draw=drawColor,line width= 1.1pt,line join=round] (110.36, 53.52) -- (118.43, 53.52);
\definecolor{fillColor}{gray}{0.20}

\path[draw=drawColor,line width= 0.4pt,line join=round,line cap=round,fill=fillColor] (146.67, 68.65) circle (  1.96);

\path[draw=drawColor,line width= 0.4pt,line join=round,line cap=round,fill=fillColor] (146.67, 61.42) circle (  1.96);

\path[draw=drawColor,line width= 0.4pt,line join=round,line cap=round,fill=fillColor] (146.67, 62.23) circle (  1.96);

\path[draw=drawColor,line width= 0.4pt,line join=round,line cap=round,fill=fillColor] (146.67, 65.12) circle (  1.96);

\path[draw=drawColor,line width= 0.4pt,line join=round,line cap=round,fill=fillColor] (146.67, 62.61) circle (  1.96);

\path[draw=drawColor,line width= 0.4pt,line join=round,line cap=round,fill=fillColor] (146.67, 63.06) circle (  1.96);

\path[draw=drawColor,line width= 0.6pt,line join=round] (146.67, 49.59) -- (146.67, 60.99);

\path[draw=drawColor,line width= 0.6pt,line join=round] (146.67, 41.96) -- (146.67, 36.13);
\definecolor{fillColor}{RGB}{248,118,109}

\path[draw=drawColor,line width= 0.6pt,line join=round,line cap=round,fill=fillColor] (142.63, 49.59) --
	(142.63, 41.96) --
	(150.70, 41.96) --
	(150.70, 49.59) --
	(142.63, 49.59) --
	cycle;

\path[draw=drawColor,line width= 1.1pt,line join=round] (142.63, 45.48) -- (150.70, 45.48);

\path[draw=drawColor,line width= 0.6pt,line join=round] ( 59.00,103.14) -- ( 59.00,134.14);

\path[draw=drawColor,line width= 0.6pt,line join=round] ( 59.00, 78.40) -- ( 59.00, 56.22);
\definecolor{fillColor}{RGB}{0,186,56}

\path[draw=drawColor,line width= 0.6pt,line join=round,line cap=round,fill=fillColor] ( 54.96,103.14) --
	( 54.96, 78.40) --
	( 63.03, 78.40) --
	( 63.03,103.14) --
	( 54.96,103.14) --
	cycle;

\path[draw=drawColor,line width= 1.1pt,line join=round] ( 54.96, 89.36) -- ( 63.03, 89.36);
\definecolor{fillColor}{gray}{0.20}

\path[draw=drawColor,line width= 0.4pt,line join=round,line cap=round,fill=fillColor] ( 91.27,103.44) circle (  1.96);

\path[draw=drawColor,line width= 0.4pt,line join=round,line cap=round,fill=fillColor] ( 91.27,106.63) circle (  1.96);

\path[draw=drawColor,line width= 0.6pt,line join=round] ( 91.27, 77.92) -- ( 91.27,100.29);

\path[draw=drawColor,line width= 0.6pt,line join=round] ( 91.27, 62.48) -- ( 91.27, 47.93);
\definecolor{fillColor}{RGB}{0,186,56}

\path[draw=drawColor,line width= 0.6pt,line join=round,line cap=round,fill=fillColor] ( 87.23, 77.92) --
	( 87.23, 62.48) --
	( 95.30, 62.48) --
	( 95.30, 77.92) --
	( 87.23, 77.92) --
	cycle;

\path[draw=drawColor,line width= 1.1pt,line join=round] ( 87.23, 69.48) -- ( 95.30, 69.48);
\definecolor{fillColor}{gray}{0.20}

\path[draw=drawColor,line width= 0.4pt,line join=round,line cap=round,fill=fillColor] (123.54, 89.02) circle (  1.96);

\path[draw=drawColor,line width= 0.4pt,line join=round,line cap=round,fill=fillColor] (123.54, 90.23) circle (  1.96);

\path[draw=drawColor,line width= 0.4pt,line join=round,line cap=round,fill=fillColor] (123.54, 90.97) circle (  1.96);

\path[draw=drawColor,line width= 0.4pt,line join=round,line cap=round,fill=fillColor] (123.54, 88.36) circle (  1.96);

\path[draw=drawColor,line width= 0.4pt,line join=round,line cap=round,fill=fillColor] (123.54, 86.88) circle (  1.96);

\path[draw=drawColor,line width= 0.6pt,line join=round] (123.54, 67.57) -- (123.54, 85.81);

\path[draw=drawColor,line width= 0.6pt,line join=round] (123.54, 55.00) -- (123.54, 44.33);
\definecolor{fillColor}{RGB}{0,186,56}

\path[draw=drawColor,line width= 0.6pt,line join=round,line cap=round,fill=fillColor] (119.51, 67.57) --
	(119.51, 55.00) --
	(127.57, 55.00) --
	(127.57, 67.57) --
	(119.51, 67.57) --
	cycle;

\path[draw=drawColor,line width= 1.1pt,line join=round] (119.51, 60.49) -- (127.57, 60.49);
\definecolor{fillColor}{gray}{0.20}

\path[draw=drawColor,line width= 0.4pt,line join=round,line cap=round,fill=fillColor] (155.81, 76.22) circle (  1.96);

\path[draw=drawColor,line width= 0.4pt,line join=round,line cap=round,fill=fillColor] (155.81, 72.59) circle (  1.96);

\path[draw=drawColor,line width= 0.4pt,line join=round,line cap=round,fill=fillColor] (155.81, 73.93) circle (  1.96);

\path[draw=drawColor,line width= 0.4pt,line join=round,line cap=round,fill=fillColor] (155.81, 74.31) circle (  1.96);

\path[draw=drawColor,line width= 0.4pt,line join=round,line cap=round,fill=fillColor] (155.81, 71.41) circle (  1.96);

\path[draw=drawColor,line width= 0.4pt,line join=round,line cap=round,fill=fillColor] (155.81, 73.22) circle (  1.96);

\path[draw=drawColor,line width= 0.4pt,line join=round,line cap=round,fill=fillColor] (155.81, 72.68) circle (  1.96);

\path[draw=drawColor,line width= 0.6pt,line join=round] (155.81, 56.95) -- (155.81, 70.57);

\path[draw=drawColor,line width= 0.6pt,line join=round] (155.81, 47.81) -- (155.81, 40.37);
\definecolor{fillColor}{RGB}{0,186,56}

\path[draw=drawColor,line width= 0.6pt,line join=round,line cap=round,fill=fillColor] (151.78, 56.95) --
	(151.78, 47.81) --
	(159.85, 47.81) --
	(159.85, 56.95) --
	(151.78, 56.95) --
	cycle;

\path[draw=drawColor,line width= 1.1pt,line join=round] (151.78, 51.55) -- (159.85, 51.55);

\path[draw=drawColor,line width= 0.6pt,line join=round] ( 68.14,103.95) -- ( 68.14,134.14);

\path[draw=drawColor,line width= 0.6pt,line join=round] ( 68.14, 78.40) -- ( 68.14, 56.22);
\definecolor{fillColor}{RGB}{97,156,255}

\path[draw=drawColor,line width= 0.6pt,line join=round,line cap=round,fill=fillColor] ( 64.11,103.95) --
	( 64.11, 78.40) --
	( 72.17, 78.40) --
	( 72.17,103.95) --
	( 64.11,103.95) --
	cycle;

\path[draw=drawColor,line width= 1.1pt,line join=round] ( 64.11, 89.45) -- ( 72.17, 89.45);
\definecolor{fillColor}{gray}{0.20}

\path[draw=drawColor,line width= 0.4pt,line join=round,line cap=round,fill=fillColor] (100.41,103.44) circle (  1.96);

\path[draw=drawColor,line width= 0.4pt,line join=round,line cap=round,fill=fillColor] (100.41,106.63) circle (  1.96);

\path[draw=drawColor,line width= 0.6pt,line join=round] (100.41, 78.00) -- (100.41,100.29);

\path[draw=drawColor,line width= 0.6pt,line join=round] (100.41, 62.50) -- (100.41, 47.93);
\definecolor{fillColor}{RGB}{97,156,255}

\path[draw=drawColor,line width= 0.6pt,line join=round,line cap=round,fill=fillColor] ( 96.38, 78.00) --
	( 96.38, 62.50) --
	(104.45, 62.50) --
	(104.45, 78.00) --
	( 96.38, 78.00) --
	cycle;

\path[draw=drawColor,line width= 1.1pt,line join=round] ( 96.38, 69.49) -- (104.45, 69.49);
\definecolor{fillColor}{gray}{0.20}

\path[draw=drawColor,line width= 0.4pt,line join=round,line cap=round,fill=fillColor] (132.68, 89.02) circle (  1.96);

\path[draw=drawColor,line width= 0.4pt,line join=round,line cap=round,fill=fillColor] (132.68, 90.23) circle (  1.96);

\path[draw=drawColor,line width= 0.4pt,line join=round,line cap=round,fill=fillColor] (132.68, 90.97) circle (  1.96);

\path[draw=drawColor,line width= 0.4pt,line join=round,line cap=round,fill=fillColor] (132.68, 88.36) circle (  1.96);

\path[draw=drawColor,line width= 0.4pt,line join=round,line cap=round,fill=fillColor] (132.68, 86.88) circle (  1.96);

\path[draw=drawColor,line width= 0.6pt,line join=round] (132.68, 67.57) -- (132.68, 85.81);

\path[draw=drawColor,line width= 0.6pt,line join=round] (132.68, 55.04) -- (132.68, 44.33);
\definecolor{fillColor}{RGB}{97,156,255}

\path[draw=drawColor,line width= 0.6pt,line join=round,line cap=round,fill=fillColor] (128.65, 67.57) --
	(128.65, 55.04) --
	(136.72, 55.04) --
	(136.72, 67.57) --
	(128.65, 67.57) --
	cycle;

\path[draw=drawColor,line width= 1.1pt,line join=round] (128.65, 60.49) -- (136.72, 60.49);
\definecolor{fillColor}{gray}{0.20}

\path[draw=drawColor,line width= 0.4pt,line join=round,line cap=round,fill=fillColor] (164.96, 76.22) circle (  1.96);

\path[draw=drawColor,line width= 0.4pt,line join=round,line cap=round,fill=fillColor] (164.96, 72.59) circle (  1.96);

\path[draw=drawColor,line width= 0.4pt,line join=round,line cap=round,fill=fillColor] (164.96, 73.93) circle (  1.96);

\path[draw=drawColor,line width= 0.4pt,line join=round,line cap=round,fill=fillColor] (164.96, 74.31) circle (  1.96);

\path[draw=drawColor,line width= 0.4pt,line join=round,line cap=round,fill=fillColor] (164.96, 71.41) circle (  1.96);

\path[draw=drawColor,line width= 0.4pt,line join=round,line cap=round,fill=fillColor] (164.96, 73.22) circle (  1.96);

\path[draw=drawColor,line width= 0.4pt,line join=round,line cap=round,fill=fillColor] (164.96, 72.68) circle (  1.96);

\path[draw=drawColor,line width= 0.6pt,line join=round] (164.96, 56.95) -- (164.96, 70.57);

\path[draw=drawColor,line width= 0.6pt,line join=round] (164.96, 47.81) -- (164.96, 40.37);
\definecolor{fillColor}{RGB}{97,156,255}

\path[draw=drawColor,line width= 0.6pt,line join=round,line cap=round,fill=fillColor] (160.92, 56.95) --
	(160.92, 47.81) --
	(168.99, 47.81) --
	(168.99, 56.95) --
	(160.92, 56.95) --
	cycle;

\path[draw=drawColor,line width= 1.1pt,line join=round] (160.92, 51.55) -- (168.99, 51.55);
\end{scope}
\begin{scope}
\path[clip] (  0.00,  0.00) rectangle (180.67,144.54);
\definecolor{drawColor}{gray}{0.30}

\node[text=drawColor,anchor=base east,inner sep=0pt, outer sep=0pt, scale=  0.97] at ( 34.68, 51.76) {40};

\node[text=drawColor,anchor=base east,inner sep=0pt, outer sep=0pt, scale=  0.97] at ( 34.68, 78.11) {60};

\node[text=drawColor,anchor=base east,inner sep=0pt, outer sep=0pt, scale=  0.97] at ( 34.68,104.46) {80};

\node[text=drawColor,anchor=base east,inner sep=0pt, outer sep=0pt, scale=  0.97] at ( 34.68,130.81) {100};
\end{scope}
\begin{scope}
\path[clip] (  0.00,  0.00) rectangle (180.67,144.54);
\definecolor{drawColor}{gray}{0.20}

\path[draw=drawColor,line width= 0.6pt,line join=round] ( 36.88, 55.10) --
	( 39.63, 55.10);

\path[draw=drawColor,line width= 0.6pt,line join=round] ( 36.88, 81.45) --
	( 39.63, 81.45);

\path[draw=drawColor,line width= 0.6pt,line join=round] ( 36.88,107.79) --
	( 39.63,107.79);

\path[draw=drawColor,line width= 0.6pt,line join=round] ( 36.88,134.14) --
	( 39.63,134.14);
\end{scope}
\begin{scope}
\path[clip] (  0.00,  0.00) rectangle (180.67,144.54);
\definecolor{drawColor}{gray}{0.20}

\path[draw=drawColor,line width= 0.6pt,line join=round] ( 59.00, 28.48) --
	( 59.00, 31.23);

\path[draw=drawColor,line width= 0.6pt,line join=round] ( 91.27, 28.48) --
	( 91.27, 31.23);

\path[draw=drawColor,line width= 0.6pt,line join=round] (123.54, 28.48) --
	(123.54, 31.23);

\path[draw=drawColor,line width= 0.6pt,line join=round] (155.81, 28.48) --
	(155.81, 31.23);
\end{scope}
\begin{scope}
\path[clip] (  0.00,  0.00) rectangle (180.67,144.54);
\definecolor{drawColor}{gray}{0.30}

\node[text=drawColor,anchor=base,inner sep=0pt, outer sep=0pt, scale=  0.97] at ( 59.00, 19.61) {(1,2]};

\node[text=drawColor,anchor=base,inner sep=0pt, outer sep=0pt, scale=  0.97] at ( 91.27, 19.61) {(2,3]};

\node[text=drawColor,anchor=base,inner sep=0pt, outer sep=0pt, scale=  0.97] at (123.54, 19.61) {(3,4]};

\node[text=drawColor,anchor=base,inner sep=0pt, outer sep=0pt, scale=  0.97] at (155.81, 19.61) {(4,5]};
\end{scope}
\begin{scope}
\path[clip] (  0.00,  0.00) rectangle (180.67,144.54);
\definecolor{drawColor}{RGB}{0,0,0}

\node[text=drawColor,anchor=base,inner sep=0pt, outer sep=0pt, scale=  0.97] at (107.40,  7.44) {Density};
\end{scope}
\begin{scope}
\path[clip] (  0.00,  0.00) rectangle (180.67,144.54);
\definecolor{drawColor}{RGB}{0,0,0}

\node[text=drawColor,rotate= 90.00,anchor=base,inner sep=0pt, outer sep=0pt, scale=  0.97] at ( 12.17, 85.13) {\% of saved crossings};
\end{scope}
\end{tikzpicture}

%% file: plots/random3/ag_time.tex
\begin{tikzpicture}[x=1pt,y=1pt]
\definecolor{fillColor}{RGB}{255,255,255}
\path[use as bounding box,fill=fillColor,fill opacity=0.00] (0,0) rectangle (180.67,144.54);
\begin{scope}
\path[clip] (  0.00,  0.00) rectangle (180.67,144.54);
\definecolor{drawColor}{RGB}{255,255,255}
\definecolor{fillColor}{RGB}{255,255,255}

\path[draw=drawColor,line width= 0.6pt,line join=round,line cap=round,fill=fillColor] (  0.00,  0.00) rectangle (180.67,144.54);
\end{scope}
\begin{scope}
\path[clip] ( 50.47, 31.23) rectangle (175.17,139.04);
\definecolor{fillColor}{gray}{0.92}

\path[fill=fillColor] ( 50.47, 31.23) rectangle (175.17,139.04);
\definecolor{drawColor}{RGB}{255,255,255}

\path[draw=drawColor,line width= 0.3pt,line join=round] ( 50.47, 36.68) --
	(175.17, 36.68);

\path[draw=drawColor,line width= 0.3pt,line join=round] ( 50.47, 66.92) --
	(175.17, 66.92);

\path[draw=drawColor,line width= 0.3pt,line join=round] ( 50.47, 97.16) --
	(175.17, 97.16);

\path[draw=drawColor,line width= 0.3pt,line join=round] ( 50.47,127.40) --
	(175.17,127.40);

\path[draw=drawColor,line width= 0.6pt,line join=round] ( 50.47, 51.80) --
	(175.17, 51.80);

\path[draw=drawColor,line width= 0.6pt,line join=round] ( 50.47, 82.04) --
	(175.17, 82.04);

\path[draw=drawColor,line width= 0.6pt,line join=round] ( 50.47,112.28) --
	(175.17,112.28);

\path[draw=drawColor,line width= 0.6pt,line join=round] ( 68.29, 31.23) --
	( 68.29,139.04);

\path[draw=drawColor,line width= 0.6pt,line join=round] ( 97.98, 31.23) --
	( 97.98,139.04);

\path[draw=drawColor,line width= 0.6pt,line join=round] (127.67, 31.23) --
	(127.67,139.04);

\path[draw=drawColor,line width= 0.6pt,line join=round] (157.36, 31.23) --
	(157.36,139.04);
\definecolor{drawColor}{gray}{0.20}
\definecolor{fillColor}{gray}{0.20}

\path[draw=drawColor,line width= 0.4pt,line join=round,line cap=round,fill=fillColor] ( 59.87, 74.85) circle (  1.96);

\path[draw=drawColor,line width= 0.4pt,line join=round,line cap=round,fill=fillColor] ( 59.87, 72.69) circle (  1.96);

\path[draw=drawColor,line width= 0.4pt,line join=round,line cap=round,fill=fillColor] ( 59.87, 72.96) circle (  1.96);

\path[draw=drawColor,line width= 0.4pt,line join=round,line cap=round,fill=fillColor] ( 59.87, 76.19) circle (  1.96);

\path[draw=drawColor,line width= 0.6pt,line join=round] ( 59.87, 57.63) -- ( 59.87, 70.90);

\path[draw=drawColor,line width= 0.6pt,line join=round] ( 59.87, 48.43) -- ( 59.87, 36.13);
\definecolor{fillColor}{RGB}{248,118,109}

\path[draw=drawColor,line width= 0.6pt,line join=round,line cap=round,fill=fillColor] ( 56.16, 57.63) --
	( 56.16, 48.43) --
	( 63.58, 48.43) --
	( 63.58, 57.63) --
	( 56.16, 57.63) --
	cycle;

\path[draw=drawColor,line width= 1.1pt,line join=round] ( 56.16, 53.01) -- ( 63.58, 53.01);
\definecolor{fillColor}{gray}{0.20}

\path[draw=drawColor,line width= 0.4pt,line join=round,line cap=round,fill=fillColor] ( 89.56, 89.40) circle (  1.96);

\path[draw=drawColor,line width= 0.4pt,line join=round,line cap=round,fill=fillColor] ( 89.56, 89.32) circle (  1.96);

\path[draw=drawColor,line width= 0.4pt,line join=round,line cap=round,fill=fillColor] ( 89.56, 91.07) circle (  1.96);

\path[draw=drawColor,line width= 0.4pt,line join=round,line cap=round,fill=fillColor] ( 89.56, 89.51) circle (  1.96);

\path[draw=drawColor,line width= 0.4pt,line join=round,line cap=round,fill=fillColor] ( 89.56, 89.17) circle (  1.96);

\path[draw=drawColor,line width= 0.4pt,line join=round,line cap=round,fill=fillColor] ( 89.56, 87.82) circle (  1.96);

\path[draw=drawColor,line width= 0.6pt,line join=round] ( 89.56, 72.50) -- ( 89.56, 87.04);

\path[draw=drawColor,line width= 0.6pt,line join=round] ( 89.56, 62.48) -- ( 89.56, 52.72);
\definecolor{fillColor}{RGB}{248,118,109}

\path[draw=drawColor,line width= 0.6pt,line join=round,line cap=round,fill=fillColor] ( 85.85, 72.50) --
	( 85.85, 62.48) --
	( 93.28, 62.48) --
	( 93.28, 72.50) --
	( 85.85, 72.50) --
	cycle;

\path[draw=drawColor,line width= 1.1pt,line join=round] ( 85.85, 66.86) -- ( 93.28, 66.86);
\definecolor{fillColor}{gray}{0.20}

\path[draw=drawColor,line width= 0.4pt,line join=round,line cap=round,fill=fillColor] (119.26,104.76) circle (  1.96);

\path[draw=drawColor,line width= 0.6pt,line join=round] (119.26, 84.57) -- (119.26,101.36);

\path[draw=drawColor,line width= 0.6pt,line join=round] (119.26, 71.66) -- (119.26, 59.73);
\definecolor{fillColor}{RGB}{248,118,109}

\path[draw=drawColor,line width= 0.6pt,line join=round,line cap=round,fill=fillColor] (115.54, 84.57) --
	(115.54, 71.66) --
	(122.97, 71.66) --
	(122.97, 84.57) --
	(115.54, 84.57) --
	cycle;

\path[draw=drawColor,line width= 1.1pt,line join=round] (115.54, 77.84) -- (122.97, 77.84);

\path[draw=drawColor,line width= 0.6pt,line join=round] (148.95, 96.05) -- (148.95,114.37);

\path[draw=drawColor,line width= 0.6pt,line join=round] (148.95, 82.90) -- (148.95, 63.81);

\path[draw=drawColor,line width= 0.6pt,line join=round,line cap=round,fill=fillColor] (145.24, 96.05) --
	(145.24, 82.90) --
	(152.66, 82.90) --
	(152.66, 96.05) --
	(145.24, 96.05) --
	cycle;

\path[draw=drawColor,line width= 1.1pt,line join=round] (145.24, 89.59) -- (152.66, 89.59);

\path[draw=drawColor,line width= 0.6pt,line join=round] ( 68.29, 66.59) -- ( 68.29, 83.55);

\path[draw=drawColor,line width= 0.6pt,line join=round] ( 68.29, 54.02) -- ( 68.29, 37.49);
\definecolor{fillColor}{RGB}{0,186,56}

\path[draw=drawColor,line width= 0.6pt,line join=round,line cap=round,fill=fillColor] ( 64.57, 66.59) --
	( 64.57, 54.02) --
	( 72.00, 54.02) --
	( 72.00, 66.59) --
	( 64.57, 66.59) --
	cycle;

\path[draw=drawColor,line width= 1.1pt,line join=round] ( 64.57, 60.44) -- ( 72.00, 60.44);
\definecolor{fillColor}{gray}{0.20}

\path[draw=drawColor,line width= 0.4pt,line join=round,line cap=round,fill=fillColor] ( 97.98,102.40) circle (  1.96);

\path[draw=drawColor,line width= 0.4pt,line join=round,line cap=round,fill=fillColor] ( 97.98,103.15) circle (  1.96);

\path[draw=drawColor,line width= 0.6pt,line join=round] ( 97.98, 83.94) -- ( 97.98,101.75);

\path[draw=drawColor,line width= 0.6pt,line join=round] ( 97.98, 72.02) -- ( 97.98, 60.19);
\definecolor{fillColor}{RGB}{0,186,56}

\path[draw=drawColor,line width= 0.6pt,line join=round,line cap=round,fill=fillColor] ( 94.27, 83.94) --
	( 94.27, 72.02) --
	(101.69, 72.02) --
	(101.69, 83.94) --
	( 94.27, 83.94) --
	cycle;

\path[draw=drawColor,line width= 1.1pt,line join=round] ( 94.27, 77.40) -- (101.69, 77.40);
\definecolor{fillColor}{gray}{0.20}

\path[draw=drawColor,line width= 0.4pt,line join=round,line cap=round,fill=fillColor] (127.67,121.32) circle (  1.96);

\path[draw=drawColor,line width= 0.6pt,line join=round] (127.67, 97.53) -- (127.67,117.47);

\path[draw=drawColor,line width= 0.6pt,line join=round] (127.67, 82.25) -- (127.67, 68.70);
\definecolor{fillColor}{RGB}{0,186,56}

\path[draw=drawColor,line width= 0.6pt,line join=round,line cap=round,fill=fillColor] (123.96, 97.53) --
	(123.96, 82.25) --
	(131.38, 82.25) --
	(131.38, 97.53) --
	(123.96, 97.53) --
	cycle;

\path[draw=drawColor,line width= 1.1pt,line join=round] (123.96, 89.59) -- (131.38, 89.59);

\path[draw=drawColor,line width= 0.6pt,line join=round] (157.36,110.97) -- (157.36,134.14);

\path[draw=drawColor,line width= 0.6pt,line join=round] (157.36, 95.44) -- (157.36, 75.02);

\path[draw=drawColor,line width= 0.6pt,line join=round,line cap=round,fill=fillColor] (153.65,110.97) --
	(153.65, 95.44) --
	(161.07, 95.44) --
	(161.07,110.97) --
	(153.65,110.97) --
	cycle;

\path[draw=drawColor,line width= 1.1pt,line join=round] (153.65,103.60) -- (161.07,103.60);

\path[draw=drawColor,line width= 0.6pt,line join=round] ( 76.70, 66.61) -- ( 76.70, 84.73);

\path[draw=drawColor,line width= 0.6pt,line join=round] ( 76.70, 53.96) -- ( 76.70, 37.17);
\definecolor{fillColor}{RGB}{97,156,255}

\path[draw=drawColor,line width= 0.6pt,line join=round,line cap=round,fill=fillColor] ( 72.99, 66.61) --
	( 72.99, 53.96) --
	( 80.41, 53.96) --
	( 80.41, 66.61) --
	( 72.99, 66.61) --
	cycle;

\path[draw=drawColor,line width= 1.1pt,line join=round] ( 72.99, 60.45) -- ( 80.41, 60.45);
\definecolor{fillColor}{gray}{0.20}

\path[draw=drawColor,line width= 0.4pt,line join=round,line cap=round,fill=fillColor] (106.39,101.96) circle (  1.96);

\path[draw=drawColor,line width= 0.4pt,line join=round,line cap=round,fill=fillColor] (106.39,102.91) circle (  1.96);

\path[draw=drawColor,line width= 0.6pt,line join=round] (106.39, 83.85) -- (106.39,101.62);

\path[draw=drawColor,line width= 0.6pt,line join=round] (106.39, 71.91) -- (106.39, 60.18);
\definecolor{fillColor}{RGB}{97,156,255}

\path[draw=drawColor,line width= 0.6pt,line join=round,line cap=round,fill=fillColor] (102.68, 83.85) --
	(102.68, 71.91) --
	(110.10, 71.91) --
	(110.10, 83.85) --
	(102.68, 83.85) --
	cycle;

\path[draw=drawColor,line width= 1.1pt,line join=round] (102.68, 77.54) -- (110.10, 77.54);
\definecolor{fillColor}{gray}{0.20}

\path[draw=drawColor,line width= 0.4pt,line join=round,line cap=round,fill=fillColor] (136.08,120.54) circle (  1.96);

\path[draw=drawColor,line width= 0.6pt,line join=round] (136.08, 97.56) -- (136.08,117.49);

\path[draw=drawColor,line width= 0.6pt,line join=round] (136.08, 82.38) -- (136.08, 68.77);
\definecolor{fillColor}{RGB}{97,156,255}

\path[draw=drawColor,line width= 0.6pt,line join=round,line cap=round,fill=fillColor] (132.37, 97.56) --
	(132.37, 82.38) --
	(139.79, 82.38) --
	(139.79, 97.56) --
	(132.37, 97.56) --
	cycle;

\path[draw=drawColor,line width= 1.1pt,line join=round] (132.37, 89.59) -- (139.79, 89.59);
\definecolor{fillColor}{gray}{0.20}

\path[draw=drawColor,line width= 0.4pt,line join=round,line cap=round,fill=fillColor] (165.77,134.12) circle (  1.96);

\path[draw=drawColor,line width= 0.6pt,line join=round] (165.77,110.95) -- (165.77,130.48);

\path[draw=drawColor,line width= 0.6pt,line join=round] (165.77, 95.54) -- (165.77, 75.01);
\definecolor{fillColor}{RGB}{97,156,255}

\path[draw=drawColor,line width= 0.6pt,line join=round,line cap=round,fill=fillColor] (162.06,110.95) --
	(162.06, 95.54) --
	(169.48, 95.54) --
	(169.48,110.95) --
	(162.06,110.95) --
	cycle;

\path[draw=drawColor,line width= 1.1pt,line join=round] (162.06,103.60) -- (169.48,103.60);
\end{scope}
\begin{scope}
\path[clip] (  0.00,  0.00) rectangle (180.67,144.54);
\definecolor{drawColor}{gray}{0.30}

\node[text=drawColor,anchor=base east,inner sep=0pt, outer sep=0pt, scale=  0.97] at ( 45.52, 48.47) {0.0001};

\node[text=drawColor,anchor=base east,inner sep=0pt, outer sep=0pt, scale=  0.97] at ( 45.52, 78.71) {0.1};

\node[text=drawColor,anchor=base east,inner sep=0pt, outer sep=0pt, scale=  0.97] at ( 45.52,108.94) {100};
\end{scope}
\begin{scope}
\path[clip] (  0.00,  0.00) rectangle (180.67,144.54);
\definecolor{drawColor}{gray}{0.20}

\path[draw=drawColor,line width= 0.6pt,line join=round] ( 47.72, 51.80) --
	( 50.47, 51.80);

\path[draw=drawColor,line width= 0.6pt,line join=round] ( 47.72, 82.04) --
	( 50.47, 82.04);

\path[draw=drawColor,line width= 0.6pt,line join=round] ( 47.72,112.28) --
	( 50.47,112.28);
\end{scope}
\begin{scope}
\path[clip] (  0.00,  0.00) rectangle (180.67,144.54);
\definecolor{drawColor}{gray}{0.20}

\path[draw=drawColor,line width= 0.6pt,line join=round] ( 68.29, 28.48) --
	( 68.29, 31.23);

\path[draw=drawColor,line width= 0.6pt,line join=round] ( 97.98, 28.48) --
	( 97.98, 31.23);

\path[draw=drawColor,line width= 0.6pt,line join=round] (127.67, 28.48) --
	(127.67, 31.23);

\path[draw=drawColor,line width= 0.6pt,line join=round] (157.36, 28.48) --
	(157.36, 31.23);
\end{scope}
\begin{scope}
\path[clip] (  0.00,  0.00) rectangle (180.67,144.54);
\definecolor{drawColor}{gray}{0.30}

\node[text=drawColor,anchor=base,inner sep=0pt, outer sep=0pt, scale=  0.97] at ( 68.29, 19.61) {(1,2]};

\node[text=drawColor,anchor=base,inner sep=0pt, outer sep=0pt, scale=  0.97] at ( 97.98, 19.61) {(2,3]};

\node[text=drawColor,anchor=base,inner sep=0pt, outer sep=0pt, scale=  0.97] at (127.67, 19.61) {(3,4]};

\node[text=drawColor,anchor=base,inner sep=0pt, outer sep=0pt, scale=  0.97] at (157.36, 19.61) {(4,5]};
\end{scope}
\begin{scope}
\path[clip] (  0.00,  0.00) rectangle (180.67,144.54);
\definecolor{drawColor}{RGB}{0,0,0}

\node[text=drawColor,anchor=base,inner sep=0pt, outer sep=0pt, scale=  0.97] at (112.82,  7.44) {Density};
\end{scope}
\begin{scope}
\path[clip] (  0.00,  0.00) rectangle (180.67,144.54);
\definecolor{drawColor}{RGB}{0,0,0}

\node[text=drawColor,rotate= 90.00,anchor=base,inner sep=0pt, outer sep=0pt, scale=  0.97] at ( 12.17, 85.13) {Time in seconds};
\end{scope}
\end{tikzpicture}

%% file: plots/legend.tex
\begin{tikzpicture}[x=1pt,y=1pt]
\begin{scope}
\definecolor{drawColor}{RGB}{255,255,255}
\definecolor{fillColor}{gray}{0.95}

\path[draw=drawColor,line width= 0.6pt,line join=round,line cap=round,fill=fillColor] ( 52.44, 11.19) rectangle ( 66.89, 25.64);
\end{scope}
\begin{scope}
\definecolor{drawColor}{gray}{0.20}

\path[draw=drawColor,line width= 0.6pt,line join=round,line cap=round] ( 59.66, 12.64) --
	( 59.66, 14.80);

\path[draw=drawColor,line width= 0.6pt,line join=round,line cap=round] ( 59.66, 22.03) --
	( 59.66, 24.20);
\definecolor{fillColor}{RGB}{248,118,109}

\path[draw=drawColor,line width= 0.6pt,line join=round,line cap=round,fill=fillColor] ( 54.24, 14.80) rectangle ( 65.08, 22.03);

\path[draw=drawColor,line width= 0.6pt,line join=round,line cap=round] ( 54.24, 18.42) --
	( 65.08, 18.42);
\end{scope}
\begin{scope}
\path[clip] (  0.00,  0.00) rectangle (180.67,144.54);
\definecolor{drawColor}{RGB}{255,255,255}
\definecolor{fillColor}{gray}{0.95}

\path[draw=drawColor,line width= 0.6pt,line join=round,line cap=round,fill=fillColor] ( 86.39, 11.19) rectangle (100.84, 25.64);
\end{scope}
\begin{scope}
\definecolor{drawColor}{gray}{0.20}

\path[draw=drawColor,line width= 0.6pt,line join=round,line cap=round] ( 93.62, 12.64) --
	( 93.62, 14.80);

\path[draw=drawColor,line width= 0.6pt,line join=round,line cap=round] ( 93.62, 22.03) --
	( 93.62, 24.20);
\definecolor{fillColor}{RGB}{0,186,56}

\path[draw=drawColor,line width= 0.6pt,line join=round,line cap=round,fill=fillColor] ( 88.20, 14.80) rectangle ( 99.04, 22.03);

\path[draw=drawColor,line width= 0.6pt,line join=round,line cap=round] ( 88.20, 18.42) --
	( 99.04, 18.42);
\end{scope}
\begin{scope}
\definecolor{drawColor}{RGB}{255,255,255}
\definecolor{fillColor}{gray}{0.95}

\path[draw=drawColor,line width= 0.6pt,line join=round,line cap=round,fill=fillColor] (146.00, 11.19) rectangle (160.46, 25.64);
\end{scope}
\begin{scope}
\definecolor{drawColor}{gray}{0.20}

\path[draw=drawColor,line width= 0.6pt,line join=round,line cap=round] (153.23, 12.64) --
	(153.23, 14.80);

\path[draw=drawColor,line width= 0.6pt,line join=round,line cap=round] (153.23, 22.03) --
	(153.23, 24.20);
\definecolor{fillColor}{RGB}{97,156,255}

\path[draw=drawColor,line width= 0.6pt,line join=round,line cap=round,fill=fillColor] (147.81, 14.80) rectangle (158.65, 22.03);

\path[draw=drawColor,line width= 0.6pt,line join=round,line cap=round] (147.81, 18.42) --
	(158.65, 18.42);
\end{scope}
\begin{scope}
\definecolor{drawColor}{RGB}{0,0,0}

\node[text=drawColor,anchor=base west,inner sep=0pt, outer sep=0pt, scale=  0.88] at ( 68.70, 15.39) {k=0};
\end{scope}
\begin{scope}
\definecolor{drawColor}{RGB}{0,0,0}

\node[text=drawColor,anchor=base west,inner sep=0pt, outer sep=0pt, scale=  0.88] at (102.65, 15.39) {k=1, w=1};
\end{scope}
\begin{scope}
\definecolor{drawColor}{RGB}{0,0,0}

\node[text=drawColor,anchor=base west,inner sep=0pt, outer sep=0pt, scale=  0.88] at (162.26, 15.39) {k=1, w=2};
\end{scope}
\end{tikzpicture}

%% file: plots/random4/ag_saved_crossings.tex
\begin{tikzpicture}[x=1pt,y=1pt]
\definecolor{fillColor}{RGB}{255,255,255}
\path[use as bounding box,fill=fillColor,fill opacity=0.00] (0,0) rectangle (180.67,144.54);
\begin{scope}
\path[clip] (  0.00,  0.00) rectangle (180.67,144.54);
\definecolor{drawColor}{RGB}{255,255,255}
\definecolor{fillColor}{RGB}{255,255,255}

\path[draw=drawColor,line width= 0.6pt,line join=round,line cap=round,fill=fillColor] ( -0.00,  0.00) rectangle (180.67,144.54);
\end{scope}
\begin{scope}
\path[clip] ( 39.63, 31.23) rectangle (175.17,139.04);
\definecolor{fillColor}{gray}{0.92}

\path[fill=fillColor] ( 39.63, 31.23) rectangle (175.18,139.04);
\definecolor{drawColor}{RGB}{255,255,255}

\path[draw=drawColor,line width= 0.3pt,line join=round] ( 39.63, 42.63) --
	(175.17, 42.63);

\path[draw=drawColor,line width= 0.3pt,line join=round] ( 39.63, 68.78) --
	(175.17, 68.78);

\path[draw=drawColor,line width= 0.3pt,line join=round] ( 39.63, 94.92) --
	(175.17, 94.92);

\path[draw=drawColor,line width= 0.3pt,line join=round] ( 39.63,121.07) --
	(175.17,121.07);

\path[draw=drawColor,line width= 0.6pt,line join=round] ( 39.63, 55.70) --
	(175.17, 55.70);

\path[draw=drawColor,line width= 0.6pt,line join=round] ( 39.63, 81.85) --
	(175.17, 81.85);

\path[draw=drawColor,line width= 0.6pt,line join=round] ( 39.63,107.99) --
	(175.17,107.99);

\path[draw=drawColor,line width= 0.6pt,line join=round] ( 39.63,134.14) --
	(175.17,134.14);

\path[draw=drawColor,line width= 0.6pt,line join=round] ( 59.00, 31.23) --
	( 59.00,139.04);

\path[draw=drawColor,line width= 0.6pt,line join=round] ( 91.27, 31.23) --
	( 91.27,139.04);

\path[draw=drawColor,line width= 0.6pt,line join=round] (123.54, 31.23) --
	(123.54,139.04);

\path[draw=drawColor,line width= 0.6pt,line join=round] (155.81, 31.23) --
	(155.81,139.04);
\definecolor{drawColor}{gray}{0.20}
\definecolor{fillColor}{gray}{0.20}

\path[draw=drawColor,line width= 0.4pt,line join=round,line cap=round,fill=fillColor] ( 49.85,134.14) circle (  1.96);

\path[draw=drawColor,line width= 0.4pt,line join=round,line cap=round,fill=fillColor] ( 49.85,134.14) circle (  1.96);

\path[draw=drawColor,line width= 0.4pt,line join=round,line cap=round,fill=fillColor] ( 49.85,134.14) circle (  1.96);

\path[draw=drawColor,line width= 0.4pt,line join=round,line cap=round,fill=fillColor] ( 49.85,134.14) circle (  1.96);

\path[draw=drawColor,line width= 0.4pt,line join=round,line cap=round,fill=fillColor] ( 49.85,134.14) circle (  1.96);

\path[draw=drawColor,line width= 0.4pt,line join=round,line cap=round,fill=fillColor] ( 49.85,134.14) circle (  1.96);

\path[draw=drawColor,line width= 0.4pt,line join=round,line cap=round,fill=fillColor] ( 49.85,134.14) circle (  1.96);

\path[draw=drawColor,line width= 0.4pt,line join=round,line cap=round,fill=fillColor] ( 49.85,134.14) circle (  1.96);

\path[draw=drawColor,line width= 0.4pt,line join=round,line cap=round,fill=fillColor] ( 49.85,134.14) circle (  1.96);

\path[draw=drawColor,line width= 0.4pt,line join=round,line cap=round,fill=fillColor] ( 49.85,134.14) circle (  1.96);

\path[draw=drawColor,line width= 0.4pt,line join=round,line cap=round,fill=fillColor] ( 49.85,134.14) circle (  1.96);

\path[draw=drawColor,line width= 0.4pt,line join=round,line cap=round,fill=fillColor] ( 49.85,134.14) circle (  1.96);

\path[draw=drawColor,line width= 0.4pt,line join=round,line cap=round,fill=fillColor] ( 49.85,134.14) circle (  1.96);

\path[draw=drawColor,line width= 0.4pt,line join=round,line cap=round,fill=fillColor] ( 49.85,134.14) circle (  1.96);

\path[draw=drawColor,line width= 0.4pt,line join=round,line cap=round,fill=fillColor] ( 49.85,134.14) circle (  1.96);

\path[draw=drawColor,line width= 0.4pt,line join=round,line cap=round,fill=fillColor] ( 49.85,134.14) circle (  1.96);

\path[draw=drawColor,line width= 0.4pt,line join=round,line cap=round,fill=fillColor] ( 49.85,134.14) circle (  1.96);

\path[draw=drawColor,line width= 0.4pt,line join=round,line cap=round,fill=fillColor] ( 49.85,134.14) circle (  1.96);

\path[draw=drawColor,line width= 0.4pt,line join=round,line cap=round,fill=fillColor] ( 49.85,134.14) circle (  1.96);

\path[draw=drawColor,line width= 0.4pt,line join=round,line cap=round,fill=fillColor] ( 49.85,134.14) circle (  1.96);

\path[draw=drawColor,line width= 0.4pt,line join=round,line cap=round,fill=fillColor] ( 49.85,134.14) circle (  1.96);

\path[draw=drawColor,line width= 0.4pt,line join=round,line cap=round,fill=fillColor] ( 49.85,134.14) circle (  1.96);

\path[draw=drawColor,line width= 0.4pt,line join=round,line cap=round,fill=fillColor] ( 49.85,134.14) circle (  1.96);

\path[draw=drawColor,line width= 0.4pt,line join=round,line cap=round,fill=fillColor] ( 49.85,134.14) circle (  1.96);

\path[draw=drawColor,line width= 0.4pt,line join=round,line cap=round,fill=fillColor] ( 49.85,134.14) circle (  1.96);

\path[draw=drawColor,line width= 0.4pt,line join=round,line cap=round,fill=fillColor] ( 49.85,134.14) circle (  1.96);

\path[draw=drawColor,line width= 0.4pt,line join=round,line cap=round,fill=fillColor] ( 49.85,134.14) circle (  1.96);

\path[draw=drawColor,line width= 0.4pt,line join=round,line cap=round,fill=fillColor] ( 49.85,134.14) circle (  1.96);

\path[draw=drawColor,line width= 0.4pt,line join=round,line cap=round,fill=fillColor] ( 49.85,134.14) circle (  1.96);

\path[draw=drawColor,line width= 0.4pt,line join=round,line cap=round,fill=fillColor] ( 49.85,134.14) circle (  1.96);

\path[draw=drawColor,line width= 0.4pt,line join=round,line cap=round,fill=fillColor] ( 49.85,134.14) circle (  1.96);

\path[draw=drawColor,line width= 0.4pt,line join=round,line cap=round,fill=fillColor] ( 49.85,134.14) circle (  1.96);

\path[draw=drawColor,line width= 0.4pt,line join=round,line cap=round,fill=fillColor] ( 49.85,134.14) circle (  1.96);

\path[draw=drawColor,line width= 0.4pt,line join=round,line cap=round,fill=fillColor] ( 49.85,134.14) circle (  1.96);

\path[draw=drawColor,line width= 0.4pt,line join=round,line cap=round,fill=fillColor] ( 49.85,134.14) circle (  1.96);

\path[draw=drawColor,line width= 0.4pt,line join=round,line cap=round,fill=fillColor] ( 49.85,134.14) circle (  1.96);

\path[draw=drawColor,line width= 0.4pt,line join=round,line cap=round,fill=fillColor] ( 49.85,134.14) circle (  1.96);

\path[draw=drawColor,line width= 0.4pt,line join=round,line cap=round,fill=fillColor] ( 49.85,134.14) circle (  1.96);

\path[draw=drawColor,line width= 0.4pt,line join=round,line cap=round,fill=fillColor] ( 49.85,134.14) circle (  1.96);

\path[draw=drawColor,line width= 0.4pt,line join=round,line cap=round,fill=fillColor] ( 49.85,134.14) circle (  1.96);

\path[draw=drawColor,line width= 0.4pt,line join=round,line cap=round,fill=fillColor] ( 49.85,134.14) circle (  1.96);

\path[draw=drawColor,line width= 0.4pt,line join=round,line cap=round,fill=fillColor] ( 49.85,134.14) circle (  1.96);

\path[draw=drawColor,line width= 0.4pt,line join=round,line cap=round,fill=fillColor] ( 49.85,134.14) circle (  1.96);

\path[draw=drawColor,line width= 0.4pt,line join=round,line cap=round,fill=fillColor] ( 49.85,134.14) circle (  1.96);

\path[draw=drawColor,line width= 0.4pt,line join=round,line cap=round,fill=fillColor] ( 49.85,134.14) circle (  1.96);

\path[draw=drawColor,line width= 0.4pt,line join=round,line cap=round,fill=fillColor] ( 49.85,134.14) circle (  1.96);

\path[draw=drawColor,line width= 0.4pt,line join=round,line cap=round,fill=fillColor] ( 49.85,134.14) circle (  1.96);

\path[draw=drawColor,line width= 0.4pt,line join=round,line cap=round,fill=fillColor] ( 49.85,134.14) circle (  1.96);

\path[draw=drawColor,line width= 0.4pt,line join=round,line cap=round,fill=fillColor] ( 49.85,134.14) circle (  1.96);

\path[draw=drawColor,line width= 0.4pt,line join=round,line cap=round,fill=fillColor] ( 49.85,134.14) circle (  1.96);

\path[draw=drawColor,line width= 0.4pt,line join=round,line cap=round,fill=fillColor] ( 49.85,134.14) circle (  1.96);

\path[draw=drawColor,line width= 0.4pt,line join=round,line cap=round,fill=fillColor] ( 49.85,134.14) circle (  1.96);

\path[draw=drawColor,line width= 0.4pt,line join=round,line cap=round,fill=fillColor] ( 49.85,134.14) circle (  1.96);

\path[draw=drawColor,line width= 0.4pt,line join=round,line cap=round,fill=fillColor] ( 49.85,134.14) circle (  1.96);

\path[draw=drawColor,line width= 0.4pt,line join=round,line cap=round,fill=fillColor] ( 49.85,134.14) circle (  1.96);

\path[draw=drawColor,line width= 0.4pt,line join=round,line cap=round,fill=fillColor] ( 49.85,134.14) circle (  1.96);

\path[draw=drawColor,line width= 0.4pt,line join=round,line cap=round,fill=fillColor] ( 49.85,134.14) circle (  1.96);

\path[draw=drawColor,line width= 0.4pt,line join=round,line cap=round,fill=fillColor] ( 49.85,134.14) circle (  1.96);

\path[draw=drawColor,line width= 0.4pt,line join=round,line cap=round,fill=fillColor] ( 49.85,134.14) circle (  1.96);

\path[draw=drawColor,line width= 0.4pt,line join=round,line cap=round,fill=fillColor] ( 49.85,134.14) circle (  1.96);

\path[draw=drawColor,line width= 0.4pt,line join=round,line cap=round,fill=fillColor] ( 49.85,134.14) circle (  1.96);

\path[draw=drawColor,line width= 0.4pt,line join=round,line cap=round,fill=fillColor] ( 49.85,134.14) circle (  1.96);

\path[draw=drawColor,line width= 0.4pt,line join=round,line cap=round,fill=fillColor] ( 49.85,134.14) circle (  1.96);

\path[draw=drawColor,line width= 0.4pt,line join=round,line cap=round,fill=fillColor] ( 49.85,134.14) circle (  1.96);

\path[draw=drawColor,line width= 0.4pt,line join=round,line cap=round,fill=fillColor] ( 49.85,134.14) circle (  1.96);

\path[draw=drawColor,line width= 0.4pt,line join=round,line cap=round,fill=fillColor] ( 49.85,134.14) circle (  1.96);

\path[draw=drawColor,line width= 0.4pt,line join=round,line cap=round,fill=fillColor] ( 49.85,134.14) circle (  1.96);

\path[draw=drawColor,line width= 0.4pt,line join=round,line cap=round,fill=fillColor] ( 49.85,134.14) circle (  1.96);

\path[draw=drawColor,line width= 0.4pt,line join=round,line cap=round,fill=fillColor] ( 49.85,134.14) circle (  1.96);

\path[draw=drawColor,line width= 0.4pt,line join=round,line cap=round,fill=fillColor] ( 49.85,134.14) circle (  1.96);

\path[draw=drawColor,line width= 0.4pt,line join=round,line cap=round,fill=fillColor] ( 49.85,134.14) circle (  1.96);

\path[draw=drawColor,line width= 0.4pt,line join=round,line cap=round,fill=fillColor] ( 49.85,134.14) circle (  1.96);

\path[draw=drawColor,line width= 0.4pt,line join=round,line cap=round,fill=fillColor] ( 49.85,134.14) circle (  1.96);

\path[draw=drawColor,line width= 0.4pt,line join=round,line cap=round,fill=fillColor] ( 49.85,134.14) circle (  1.96);

\path[draw=drawColor,line width= 0.4pt,line join=round,line cap=round,fill=fillColor] ( 49.85,134.14) circle (  1.96);

\path[draw=drawColor,line width= 0.4pt,line join=round,line cap=round,fill=fillColor] ( 49.85,134.14) circle (  1.96);

\path[draw=drawColor,line width= 0.6pt,line join=round] ( 49.85, 91.34) -- ( 49.85,128.20);

\path[draw=drawColor,line width= 0.6pt,line join=round] ( 49.85, 65.97) -- ( 49.85, 45.11);
\definecolor{fillColor}{RGB}{248,118,109}

\path[draw=drawColor,line width= 0.6pt,line join=round,line cap=round,fill=fillColor] ( 45.82, 91.34) --
	( 45.82, 65.97) --
	( 53.89, 65.97) --
	( 53.89, 91.34) --
	( 45.82, 91.34) --
	cycle;

\path[draw=drawColor,line width= 1.1pt,line join=round] ( 45.82, 76.86) -- ( 53.89, 76.86);
\definecolor{fillColor}{gray}{0.20}

\path[draw=drawColor,line width= 0.4pt,line join=round,line cap=round,fill=fillColor] ( 82.12, 95.95) circle (  1.96);

\path[draw=drawColor,line width= 0.4pt,line join=round,line cap=round,fill=fillColor] ( 82.12, 93.99) circle (  1.96);

\path[draw=drawColor,line width= 0.4pt,line join=round,line cap=round,fill=fillColor] ( 82.12, 96.66) circle (  1.96);

\path[draw=drawColor,line width= 0.4pt,line join=round,line cap=round,fill=fillColor] ( 82.12, 92.60) circle (  1.96);

\path[draw=drawColor,line width= 0.4pt,line join=round,line cap=round,fill=fillColor] ( 82.12, 95.28) circle (  1.96);

\path[draw=drawColor,line width= 0.6pt,line join=round] ( 82.12, 67.35) -- ( 82.12, 89.84);

\path[draw=drawColor,line width= 0.6pt,line join=round] ( 82.12, 52.31) -- ( 82.12, 39.60);
\definecolor{fillColor}{RGB}{248,118,109}

\path[draw=drawColor,line width= 0.6pt,line join=round,line cap=round,fill=fillColor] ( 78.09, 67.35) --
	( 78.09, 52.31) --
	( 86.16, 52.31) --
	( 86.16, 67.35) --
	( 78.09, 67.35) --
	cycle;

\path[draw=drawColor,line width= 1.1pt,line join=round] ( 78.09, 58.65) -- ( 86.16, 58.65);
\definecolor{fillColor}{gray}{0.20}

\path[draw=drawColor,line width= 0.4pt,line join=round,line cap=round,fill=fillColor] (114.40, 75.82) circle (  1.96);

\path[draw=drawColor,line width= 0.4pt,line join=round,line cap=round,fill=fillColor] (114.40, 75.68) circle (  1.96);

\path[draw=drawColor,line width= 0.4pt,line join=round,line cap=round,fill=fillColor] (114.40, 75.82) circle (  1.96);

\path[draw=drawColor,line width= 0.4pt,line join=round,line cap=round,fill=fillColor] (114.40, 78.28) circle (  1.96);

\path[draw=drawColor,line width= 0.4pt,line join=round,line cap=round,fill=fillColor] (114.40, 74.99) circle (  1.96);

\path[draw=drawColor,line width= 0.6pt,line join=round] (114.40, 56.62) -- (114.40, 74.68);

\path[draw=drawColor,line width= 0.6pt,line join=round] (114.40, 44.49) -- (114.40, 36.13);
\definecolor{fillColor}{RGB}{248,118,109}

\path[draw=drawColor,line width= 0.6pt,line join=round,line cap=round,fill=fillColor] (110.36, 56.62) --
	(110.36, 44.49) --
	(118.43, 44.49) --
	(118.43, 56.62) --
	(110.36, 56.62) --
	cycle;

\path[draw=drawColor,line width= 1.1pt,line join=round] (110.36, 49.58) -- (118.43, 49.58);
\definecolor{fillColor}{gray}{0.20}

\path[draw=drawColor,line width= 0.4pt,line join=round,line cap=round,fill=fillColor] (146.67, 54.50) circle (  1.96);

\path[draw=drawColor,line width= 0.6pt,line join=round] (146.67, 44.37) -- (146.67, 49.38);

\path[draw=drawColor,line width= 0.6pt,line join=round] (146.67, 39.04) -- (146.67, 36.15);
\definecolor{fillColor}{RGB}{248,118,109}

\path[draw=drawColor,line width= 0.6pt,line join=round,line cap=round,fill=fillColor] (142.63, 44.37) --
	(142.63, 39.04) --
	(150.70, 39.04) --
	(150.70, 44.37) --
	(142.63, 44.37) --
	cycle;

\path[draw=drawColor,line width= 1.1pt,line join=round] (142.63, 40.96) -- (150.70, 40.96);

\path[draw=drawColor,line width= 0.6pt,line join=round] ( 59.00, 99.54) -- ( 59.00,134.14);

\path[draw=drawColor,line width= 0.6pt,line join=round] ( 59.00, 73.93) -- ( 59.00, 50.89);
\definecolor{fillColor}{RGB}{0,186,56}

\path[draw=drawColor,line width= 0.6pt,line join=round,line cap=round,fill=fillColor] ( 54.96, 99.54) --
	( 54.96, 73.93) --
	( 63.03, 73.93) --
	( 63.03, 99.54) --
	( 54.96, 99.54) --
	cycle;

\path[draw=drawColor,line width= 1.1pt,line join=round] ( 54.96, 84.70) -- ( 63.03, 84.70);

\path[draw=drawColor,line width= 0.6pt,line join=round] ( 91.27, 75.39) -- ( 91.27, 97.72);

\path[draw=drawColor,line width= 0.6pt,line join=round] ( 91.27, 58.35) -- ( 91.27, 44.77);

\path[draw=drawColor,line width= 0.6pt,line join=round,line cap=round,fill=fillColor] ( 87.23, 75.39) --
	( 87.23, 58.35) --
	( 95.30, 58.35) --
	( 95.30, 75.39) --
	( 87.23, 75.39) --
	cycle;

\path[draw=drawColor,line width= 1.1pt,line join=round] ( 87.23, 65.73) -- ( 95.30, 65.73);
\definecolor{fillColor}{gray}{0.20}

\path[draw=drawColor,line width= 0.4pt,line join=round,line cap=round,fill=fillColor] (123.54, 88.14) circle (  1.96);

\path[draw=drawColor,line width= 0.4pt,line join=round,line cap=round,fill=fillColor] (123.54, 84.52) circle (  1.96);

\path[draw=drawColor,line width= 0.4pt,line join=round,line cap=round,fill=fillColor] (123.54, 86.92) circle (  1.96);

\path[draw=drawColor,line width= 0.6pt,line join=round] (123.54, 63.86) -- (123.54, 83.65);

\path[draw=drawColor,line width= 0.6pt,line join=round] (123.54, 50.28) -- (123.54, 40.19);
\definecolor{fillColor}{RGB}{0,186,56}

\path[draw=drawColor,line width= 0.6pt,line join=round,line cap=round,fill=fillColor] (119.51, 63.86) --
	(119.51, 50.28) --
	(127.57, 50.28) --
	(127.57, 63.86) --
	(119.51, 63.86) --
	cycle;

\path[draw=drawColor,line width= 1.1pt,line join=round] (119.51, 55.92) -- (127.57, 55.92);

\path[draw=drawColor,line width= 0.6pt,line join=round] (155.81, 50.39) -- (155.81, 57.55);

\path[draw=drawColor,line width= 0.6pt,line join=round] (155.81, 44.38) -- (155.81, 40.14);

\path[draw=drawColor,line width= 0.6pt,line join=round,line cap=round,fill=fillColor] (151.78, 50.39) --
	(151.78, 44.38) --
	(159.85, 44.38) --
	(159.85, 50.39) --
	(151.78, 50.39) --
	cycle;

\path[draw=drawColor,line width= 1.1pt,line join=round] (151.78, 46.85) -- (159.85, 46.85);

\path[draw=drawColor,line width= 0.6pt,line join=round] ( 68.14,100.04) -- ( 68.14,134.14);

\path[draw=drawColor,line width= 0.6pt,line join=round] ( 68.14, 73.96) -- ( 68.14, 50.89);
\definecolor{fillColor}{RGB}{97,156,255}

\path[draw=drawColor,line width= 0.6pt,line join=round,line cap=round,fill=fillColor] ( 64.11,100.04) --
	( 64.11, 73.96) --
	( 72.17, 73.96) --
	( 72.17,100.04) --
	( 64.11,100.04) --
	cycle;

\path[draw=drawColor,line width= 1.1pt,line join=round] ( 64.11, 84.77) -- ( 72.17, 84.77);

\path[draw=drawColor,line width= 0.6pt,line join=round] (100.41, 75.42) -- (100.41, 97.72);

\path[draw=drawColor,line width= 0.6pt,line join=round] (100.41, 58.36) -- (100.41, 44.77);

\path[draw=drawColor,line width= 0.6pt,line join=round,line cap=round,fill=fillColor] ( 96.38, 75.42) --
	( 96.38, 58.36) --
	(104.45, 58.36) --
	(104.45, 75.42) --
	( 96.38, 75.42) --
	cycle;

\path[draw=drawColor,line width= 1.1pt,line join=round] ( 96.38, 65.74) -- (104.45, 65.74);
\definecolor{fillColor}{gray}{0.20}

\path[draw=drawColor,line width= 0.4pt,line join=round,line cap=round,fill=fillColor] (132.68, 88.14) circle (  1.96);

\path[draw=drawColor,line width= 0.4pt,line join=round,line cap=round,fill=fillColor] (132.68, 84.52) circle (  1.96);

\path[draw=drawColor,line width= 0.4pt,line join=round,line cap=round,fill=fillColor] (132.68, 88.07) circle (  1.96);

\path[draw=drawColor,line width= 0.6pt,line join=round] (132.68, 63.86) -- (132.68, 83.65);

\path[draw=drawColor,line width= 0.6pt,line join=round] (132.68, 50.28) -- (132.68, 40.19);
\definecolor{fillColor}{RGB}{97,156,255}

\path[draw=drawColor,line width= 0.6pt,line join=round,line cap=round,fill=fillColor] (128.65, 63.86) --
	(128.65, 50.28) --
	(136.72, 50.28) --
	(136.72, 63.86) --
	(128.65, 63.86) --
	cycle;

\path[draw=drawColor,line width= 1.1pt,line join=round] (128.65, 55.92) -- (136.72, 55.92);

\path[draw=drawColor,line width= 0.6pt,line join=round] (164.96, 50.39) -- (164.96, 57.55);

\path[draw=drawColor,line width= 0.6pt,line join=round] (164.96, 44.38) -- (164.96, 40.14);

\path[draw=drawColor,line width= 0.6pt,line join=round,line cap=round,fill=fillColor] (160.92, 50.39) --
	(160.92, 44.38) --
	(168.99, 44.38) --
	(168.99, 50.39) --
	(160.92, 50.39) --
	cycle;

\path[draw=drawColor,line width= 1.1pt,line join=round] (160.92, 46.85) -- (168.99, 46.85);
\end{scope}
\begin{scope}
\path[clip] (  0.00,  0.00) rectangle (180.67,144.54);
\definecolor{drawColor}{gray}{0.30}

\node[text=drawColor,anchor=base east,inner sep=0pt, outer sep=0pt, scale=  0.97] at ( 34.68, 52.37) {40};

\node[text=drawColor,anchor=base east,inner sep=0pt, outer sep=0pt, scale=  0.97] at ( 34.68, 78.52) {60};

\node[text=drawColor,anchor=base east,inner sep=0pt, outer sep=0pt, scale=  0.97] at ( 34.68,104.66) {80};

\node[text=drawColor,anchor=base east,inner sep=0pt, outer sep=0pt, scale=  0.97] at ( 34.68,130.81) {100};
\end{scope}
\begin{scope}
\path[clip] (  0.00,  0.00) rectangle (180.67,144.54);
\definecolor{drawColor}{gray}{0.20}

\path[draw=drawColor,line width= 0.6pt,line join=round] ( 36.88, 55.70) --
	( 39.63, 55.70);

\path[draw=drawColor,line width= 0.6pt,line join=round] ( 36.88, 81.85) --
	( 39.63, 81.85);

\path[draw=drawColor,line width= 0.6pt,line join=round] ( 36.88,107.99) --
	( 39.63,107.99);

\path[draw=drawColor,line width= 0.6pt,line join=round] ( 36.88,134.14) --
	( 39.63,134.14);
\end{scope}
\begin{scope}
\path[clip] (  0.00,  0.00) rectangle (180.67,144.54);
\definecolor{drawColor}{gray}{0.20}

\path[draw=drawColor,line width= 0.6pt,line join=round] ( 59.00, 28.48) --
	( 59.00, 31.23);

\path[draw=drawColor,line width= 0.6pt,line join=round] ( 91.27, 28.48) --
	( 91.27, 31.23);

\path[draw=drawColor,line width= 0.6pt,line join=round] (123.54, 28.48) --
	(123.54, 31.23);

\path[draw=drawColor,line width= 0.6pt,line join=round] (155.81, 28.48) --
	(155.81, 31.23);
\end{scope}
\begin{scope}
\path[clip] (  0.00,  0.00) rectangle (180.67,144.54);
\definecolor{drawColor}{gray}{0.30}

\node[text=drawColor,anchor=base,inner sep=0pt, outer sep=0pt, scale=  0.97] at ( 59.00, 19.61) {(1,2]};

\node[text=drawColor,anchor=base,inner sep=0pt, outer sep=0pt, scale=  0.97] at ( 91.27, 19.61) {(2,3]};

\node[text=drawColor,anchor=base,inner sep=0pt, outer sep=0pt, scale=  0.97] at (123.54, 19.61) {(3,4]};

\node[text=drawColor,anchor=base,inner sep=0pt, outer sep=0pt, scale=  0.97] at (155.81, 19.61) {(4,5]};
\end{scope}
\begin{scope}
\path[clip] (  0.00,  0.00) rectangle (180.67,144.54);
\definecolor{drawColor}{RGB}{0,0,0}

\node[text=drawColor,anchor=base,inner sep=0pt, outer sep=0pt, scale=  0.97] at (107.40,  7.44) {Density};
\end{scope}
\begin{scope}
\path[clip] (  0.00,  0.00) rectangle (180.67,144.54);
\definecolor{drawColor}{RGB}{0,0,0}

\node[text=drawColor,rotate= 90.00,anchor=base,inner sep=0pt, outer sep=0pt, scale=  0.97] at ( 12.17, 85.13) {\% of saved crossings};
\end{scope}
\end{tikzpicture}

%% file: plots/random4/ag_time.tex
\begin{tikzpicture}[x=1pt,y=1pt]
\definecolor{fillColor}{RGB}{255,255,255}
\path[use as bounding box,fill=fillColor,fill opacity=0.00] (0,0) rectangle (180.67,144.54);
\begin{scope}
\path[clip] (  0.00,  0.00) rectangle (180.67,144.54);
\definecolor{drawColor}{RGB}{255,255,255}
\definecolor{fillColor}{RGB}{255,255,255}

\path[draw=drawColor,line width= 0.6pt,line join=round,line cap=round,fill=fillColor] (  0.00,  0.00) rectangle (180.67,144.54);
\end{scope}
\begin{scope}
\path[clip] ( 45.63, 31.23) rectangle (175.17,139.04);
\definecolor{fillColor}{gray}{0.92}

\path[fill=fillColor] ( 45.63, 31.23) rectangle (175.17,139.04);
\definecolor{drawColor}{RGB}{255,255,255}

\path[draw=drawColor,line width= 0.3pt,line join=round] ( 45.63, 45.40) --
	(175.17, 45.40);

\path[draw=drawColor,line width= 0.3pt,line join=round] ( 45.63, 75.40) --
	(175.17, 75.40);

\path[draw=drawColor,line width= 0.3pt,line join=round] ( 45.63,105.39) --
	(175.17,105.39);

\path[draw=drawColor,line width= 0.3pt,line join=round] ( 45.63,135.39) --
	(175.17,135.39);

\path[draw=drawColor,line width= 0.6pt,line join=round] ( 45.63, 60.40) --
	(175.17, 60.40);

\path[draw=drawColor,line width= 0.6pt,line join=round] ( 45.63, 90.39) --
	(175.17, 90.39);

\path[draw=drawColor,line width= 0.6pt,line join=round] ( 45.63,120.39) --
	(175.17,120.39);

\path[draw=drawColor,line width= 0.6pt,line join=round] ( 64.14, 31.23) --
	( 64.14,139.04);

\path[draw=drawColor,line width= 0.6pt,line join=round] ( 94.98, 31.23) --
	( 94.98,139.04);

\path[draw=drawColor,line width= 0.6pt,line join=round] (125.83, 31.23) --
	(125.83,139.04);

\path[draw=drawColor,line width= 0.6pt,line join=round] (156.67, 31.23) --
	(156.67,139.04);
\definecolor{drawColor}{gray}{0.20}
\definecolor{fillColor}{gray}{0.20}

\path[draw=drawColor,line width= 0.4pt,line join=round,line cap=round,fill=fillColor] ( 55.40, 72.33) circle (  1.96);

\path[draw=drawColor,line width= 0.4pt,line join=round,line cap=round,fill=fillColor] ( 55.40, 73.89) circle (  1.96);

\path[draw=drawColor,line width= 0.4pt,line join=round,line cap=round,fill=fillColor] ( 55.40, 72.68) circle (  1.96);

\path[draw=drawColor,line width= 0.4pt,line join=round,line cap=round,fill=fillColor] ( 55.40, 73.75) circle (  1.96);

\path[draw=drawColor,line width= 0.4pt,line join=round,line cap=round,fill=fillColor] ( 55.40, 72.21) circle (  1.96);

\path[draw=drawColor,line width= 0.4pt,line join=round,line cap=round,fill=fillColor] ( 55.40, 73.26) circle (  1.96);

\path[draw=drawColor,line width= 0.6pt,line join=round] ( 55.40, 57.74) -- ( 55.40, 71.47);

\path[draw=drawColor,line width= 0.6pt,line join=round] ( 55.40, 48.16) -- ( 55.40, 36.13);
\definecolor{fillColor}{RGB}{248,118,109}

\path[draw=drawColor,line width= 0.6pt,line join=round,line cap=round,fill=fillColor] ( 51.54, 57.74) --
	( 51.54, 48.16) --
	( 59.25, 48.16) --
	( 59.25, 57.74) --
	( 51.54, 57.74) --
	cycle;

\path[draw=drawColor,line width= 1.1pt,line join=round] ( 51.54, 53.05) -- ( 59.25, 53.05);
\definecolor{fillColor}{gray}{0.20}

\path[draw=drawColor,line width= 0.4pt,line join=round,line cap=round,fill=fillColor] ( 86.24, 95.05) circle (  1.96);

\path[draw=drawColor,line width= 0.4pt,line join=round,line cap=round,fill=fillColor] ( 86.24, 95.22) circle (  1.96);

\path[draw=drawColor,line width= 0.6pt,line join=round] ( 86.24, 74.92) -- ( 86.24, 92.33);

\path[draw=drawColor,line width= 0.6pt,line join=round] ( 86.24, 62.18) -- ( 86.24, 52.39);
\definecolor{fillColor}{RGB}{248,118,109}

\path[draw=drawColor,line width= 0.6pt,line join=round,line cap=round,fill=fillColor] ( 82.39, 74.92) --
	( 82.39, 62.18) --
	( 90.10, 62.18) --
	( 90.10, 74.92) --
	( 82.39, 74.92) --
	cycle;

\path[draw=drawColor,line width= 1.1pt,line join=round] ( 82.39, 68.17) -- ( 90.10, 68.17);

\path[draw=drawColor,line width= 0.6pt,line join=round] (117.09, 89.87) -- (117.09,112.19);

\path[draw=drawColor,line width= 0.6pt,line join=round] (117.09, 73.65) -- (117.09, 57.73);

\path[draw=drawColor,line width= 0.6pt,line join=round,line cap=round,fill=fillColor] (113.23, 89.87) --
	(113.23, 73.65) --
	(120.94, 73.65) --
	(120.94, 89.87) --
	(113.23, 89.87) --
	cycle;

\path[draw=drawColor,line width= 1.1pt,line join=round] (113.23, 81.90) -- (120.94, 81.90);

\path[draw=drawColor,line width= 0.6pt,line join=round] (147.93,103.60) -- (147.93,113.52);

\path[draw=drawColor,line width= 0.6pt,line join=round] (147.93, 93.10) -- (147.93, 78.49);

\path[draw=drawColor,line width= 0.6pt,line join=round,line cap=round,fill=fillColor] (144.07,103.60) --
	(144.07, 93.10) --
	(151.79, 93.10) --
	(151.79,103.60) --
	(144.07,103.60) --
	cycle;

\path[draw=drawColor,line width= 1.1pt,line join=round] (144.07, 99.08) -- (151.79, 99.08);

\path[draw=drawColor,line width= 0.6pt,line join=round] ( 64.14, 67.33) -- ( 64.14, 84.33);

\path[draw=drawColor,line width= 0.6pt,line join=round] ( 64.14, 54.25) -- ( 64.14, 37.40);
\definecolor{fillColor}{RGB}{0,186,56}

\path[draw=drawColor,line width= 0.6pt,line join=round,line cap=round,fill=fillColor] ( 60.28, 67.33) --
	( 60.28, 54.25) --
	( 67.99, 54.25) --
	( 67.99, 67.33) --
	( 60.28, 67.33) --
	cycle;

\path[draw=drawColor,line width= 1.1pt,line join=round] ( 60.28, 61.13) -- ( 67.99, 61.13);
\definecolor{fillColor}{gray}{0.20}

\path[draw=drawColor,line width= 0.4pt,line join=round,line cap=round,fill=fillColor] ( 94.98,110.32) circle (  1.96);

\path[draw=drawColor,line width= 0.6pt,line join=round] ( 94.98, 86.78) -- ( 94.98,107.81);

\path[draw=drawColor,line width= 0.6pt,line join=round] ( 94.98, 72.04) -- ( 94.98, 59.28);
\definecolor{fillColor}{RGB}{0,186,56}

\path[draw=drawColor,line width= 0.6pt,line join=round,line cap=round,fill=fillColor] ( 91.13, 86.78) --
	( 91.13, 72.04) --
	( 98.84, 72.04) --
	( 98.84, 86.78) --
	( 91.13, 86.78) --
	cycle;

\path[draw=drawColor,line width= 1.1pt,line join=round] ( 91.13, 79.39) -- ( 98.84, 79.39);

\path[draw=drawColor,line width= 0.6pt,line join=round] (125.83,104.52) -- (125.83,130.99);

\path[draw=drawColor,line width= 0.6pt,line join=round] (125.83, 84.72) -- (125.83, 67.23);

\path[draw=drawColor,line width= 0.6pt,line join=round,line cap=round,fill=fillColor] (121.97,104.52) --
	(121.97, 84.72) --
	(129.68, 84.72) --
	(129.68,104.52) --
	(121.97,104.52) --
	cycle;

\path[draw=drawColor,line width= 1.1pt,line join=round] (121.97, 94.89) -- (129.68, 94.89);

\path[draw=drawColor,line width= 0.6pt,line join=round] (156.67,120.96) -- (156.67,134.14);

\path[draw=drawColor,line width= 0.6pt,line join=round] (156.67,108.30) -- (156.67, 90.49);

\path[draw=drawColor,line width= 0.6pt,line join=round,line cap=round,fill=fillColor] (152.81,120.96) --
	(152.81,108.30) --
	(160.52,108.30) --
	(160.52,120.96) --
	(152.81,120.96) --
	cycle;

\path[draw=drawColor,line width= 1.1pt,line join=round] (152.81,115.20) -- (160.52,115.20);

\path[draw=drawColor,line width= 0.6pt,line join=round] ( 72.88, 67.50) -- ( 72.88, 84.38);

\path[draw=drawColor,line width= 0.6pt,line join=round] ( 72.88, 54.27) -- ( 72.88, 36.37);
\definecolor{fillColor}{RGB}{97,156,255}

\path[draw=drawColor,line width= 0.6pt,line join=round,line cap=round,fill=fillColor] ( 69.02, 67.50) --
	( 69.02, 54.27) --
	( 76.73, 54.27) --
	( 76.73, 67.50) --
	( 69.02, 67.50) --
	cycle;

\path[draw=drawColor,line width= 1.1pt,line join=round] ( 69.02, 61.10) -- ( 76.73, 61.10);
\definecolor{fillColor}{gray}{0.20}

\path[draw=drawColor,line width= 0.4pt,line join=round,line cap=round,fill=fillColor] (103.72,110.33) circle (  1.96);

\path[draw=drawColor,line width= 0.6pt,line join=round] (103.72, 86.75) -- (103.72,108.00);

\path[draw=drawColor,line width= 0.6pt,line join=round] (103.72, 72.07) -- (103.72, 59.19);
\definecolor{fillColor}{RGB}{97,156,255}

\path[draw=drawColor,line width= 0.6pt,line join=round,line cap=round,fill=fillColor] ( 99.87, 86.75) --
	( 99.87, 72.07) --
	(107.58, 72.07) --
	(107.58, 86.75) --
	( 99.87, 86.75) --
	cycle;

\path[draw=drawColor,line width= 1.1pt,line join=round] ( 99.87, 79.47) -- (107.58, 79.47);

\path[draw=drawColor,line width= 0.6pt,line join=round] (134.56,104.40) -- (134.56,130.75);

\path[draw=drawColor,line width= 0.6pt,line join=round] (134.56, 84.89) -- (134.56, 67.39);

\path[draw=drawColor,line width= 0.6pt,line join=round,line cap=round,fill=fillColor] (130.71,104.40) --
	(130.71, 84.89) --
	(138.42, 84.89) --
	(138.42,104.40) --
	(130.71,104.40) --
	cycle;

\path[draw=drawColor,line width= 1.1pt,line join=round] (130.71, 95.02) -- (138.42, 95.02);

\path[draw=drawColor,line width= 0.6pt,line join=round] (165.41,120.56) -- (165.41,132.70);

\path[draw=drawColor,line width= 0.6pt,line join=round] (165.41,108.29) -- (165.41, 90.49);

\path[draw=drawColor,line width= 0.6pt,line join=round,line cap=round,fill=fillColor] (161.55,120.56) --
	(161.55,108.29) --
	(169.26,108.29) --
	(169.26,120.56) --
	(161.55,120.56) --
	cycle;

\path[draw=drawColor,line width= 1.1pt,line join=round] (161.55,114.84) -- (169.26,114.84);
\end{scope}
\begin{scope}
\path[clip] (  0.00,  0.00) rectangle (180.67,144.54);
\definecolor{drawColor}{gray}{0.30}

\node[text=drawColor,anchor=base east,inner sep=0pt, outer sep=0pt, scale=  0.97] at ( 40.68, 57.07) {0.001};

\node[text=drawColor,anchor=base east,inner sep=0pt, outer sep=0pt, scale=  0.97] at ( 40.68, 87.06) {1};

\node[text=drawColor,anchor=base east,inner sep=0pt, outer sep=0pt, scale=  0.97] at ( 40.68,117.06) {1000};
\end{scope}
\begin{scope}
\path[clip] (  0.00,  0.00) rectangle (180.67,144.54);
\definecolor{drawColor}{gray}{0.20}

\path[draw=drawColor,line width= 0.6pt,line join=round] ( 42.88, 60.40) --
	( 45.63, 60.40);

\path[draw=drawColor,line width= 0.6pt,line join=round] ( 42.88, 90.39) --
	( 45.63, 90.39);

\path[draw=drawColor,line width= 0.6pt,line join=round] ( 42.88,120.39) --
	( 45.63,120.39);
\end{scope}
\begin{scope}
\path[clip] (  0.00,  0.00) rectangle (180.67,144.54);
\definecolor{drawColor}{gray}{0.20}

\path[draw=drawColor,line width= 0.6pt,line join=round] ( 64.14, 28.48) --
	( 64.14, 31.23);

\path[draw=drawColor,line width= 0.6pt,line join=round] ( 94.98, 28.48) --
	( 94.98, 31.23);

\path[draw=drawColor,line width= 0.6pt,line join=round] (125.83, 28.48) --
	(125.83, 31.23);

\path[draw=drawColor,line width= 0.6pt,line join=round] (156.67, 28.48) --
	(156.67, 31.23);
\end{scope}
\begin{scope}
\path[clip] (  0.00,  0.00) rectangle (180.67,144.54);
\definecolor{drawColor}{gray}{0.30}

\node[text=drawColor,anchor=base,inner sep=0pt, outer sep=0pt, scale=  0.97] at ( 64.14, 19.61) {(1,2]};

\node[text=drawColor,anchor=base,inner sep=0pt, outer sep=0pt, scale=  0.97] at ( 94.98, 19.61) {(2,3]};

\node[text=drawColor,anchor=base,inner sep=0pt, outer sep=0pt, scale=  0.97] at (125.83, 19.61) {(3,4]};

\node[text=drawColor,anchor=base,inner sep=0pt, outer sep=0pt, scale=  0.97] at (156.67, 19.61) {(4,5]};
\end{scope}
\begin{scope}
\path[clip] (  0.00,  0.00) rectangle (180.67,144.54);
\definecolor{drawColor}{RGB}{0,0,0}

\node[text=drawColor,anchor=base,inner sep=0pt, outer sep=0pt, scale=  0.97] at (110.40,  7.44) {Density};
\end{scope}
\begin{scope}
\path[clip] (  0.00,  0.00) rectangle (180.67,144.54);
\definecolor{drawColor}{RGB}{0,0,0}

\node[text=drawColor,rotate= 90.00,anchor=base,inner sep=0pt, outer sep=0pt, scale=  0.97] at ( 12.17, 85.13) {Time in seconds};
\end{scope}
\end{tikzpicture}